\documentclass[a4paper]{article}

\usepackage[utf8]{inputenc}
\usepackage[left=30mm, right=30mm, top=30mm, bottom=30mm]{geometry}
\usepackage{amssymb}
\usepackage{amsthm}
\usepackage{mathtools}
\usepackage{xparse}
\usepackage{todonotes}
\usepackage{pgfplots}
\usepackage{wasysym} % for \clock
\usepackage{fontawesome} % for \faClockO
\usepackage[titlenumbered]{algorithm2e}
\usepackage[colorlinks,citecolor=blue,linkcolor=blue,hypertexnames=false]{hyperref}

\pgfplotsset{compat=1.15}

\SetKwInput{AlgoInput}{\textbf{Input}}
\SetKwInput{AlgoOutput}{\textbf{Output}}
\SetKwInput{AlgoComplexity}{\textbf{Complexity}}
\DontPrintSemicolon

\newtheorem{theorem}{Theorem}
\newtheorem{proposition}[theorem]{Proposition}
\newtheorem{hypothesis}[theorem]{Hypothesis}

\DeclareDocumentCommand\R{}{\mathbb{R}}
\DeclareDocumentCommand\Z{}{\mathbb{Z}}
\DeclareDocumentCommand\conv{o}{\operatorname{conv}\IfValueTF{#1}{\left(#1\right)}{}}
\DeclareDocumentCommand\aff{o}{\operatorname{aff}\IfValueTF{#1}{\left(#1\right)}{}}
\DeclareDocumentCommand\lin{o}{\operatorname{span}\IfValueTF{#1}{\left(#1\right)}{}}
\DeclareDocumentCommand\rows{o}{\operatorname{rows}\IfValueTF{#1}{\left(#1\right)}{}}
\DeclareDocumentCommand\transpose{m}{#1^{\intercal}}
\DeclareDocumentCommand\zerovec{o}{\IfNoValueTF{#1}{\mathbb{O}}{\mathbb{O}_{#1}}}
\DeclareDocumentCommand\SCIP{}{\texttt{SCIP}}
\DeclareDocumentCommand\IPO{}{\texttt{IPO}}

\title{Face Dimensions of General-Purpose Cutting Planes for Mixed-Integer Linear Programs}
\author{Matthias Walter\footnote{Department of Applied Mathematics, University of Twente, The Netherlands, \texttt{m.walter@utwente.nl}}}
\date{\small\today}

\begin{document}

\maketitle

\begin{abstract}
  Cutting planes are a key ingredient to successfully solve mixed-integer linear programs.
  For specific problems, their strength is often theoretically assessed by showing that they are facet-defining for the corresponding mixed-integer hull.
  In this paper we experimentally investigate the dimensions of faces induced by general-purpose cutting planes generated by a state-of-the-art solver.
  Therefore, we relate the dimension of each cutting plane to its impact in a branch-and-bound algorithm.
\end{abstract}

\section{Introduction}
\label{sec_introduction}

We consider the mixed-integer program
\begin{subequations}
  \label{model_mip}
  \begin{alignat}{7}
    & \text{max }   & \transpose{c}x \label{model_mip_obj} \\
    & \text{s.t. }
      & Ax &\leq b \label{model_mip_cons} \\
    & & x_i &\in \Z &\quad& \forall i \in I \label{model_mip_int}
  \end{alignat}
\end{subequations}
for a matrix $A \in \R^{m \times n}$, vectors $c \in \R^n$ and $b \in \R^m$ and a subset $I \subseteq \{1,2,\dotsc,n\}$ of integer variables.
Let $P \coloneqq \conv \{ x \in \R^n : x \text{ satisfies~\eqref{model_mip_cons} and~\eqref{model_mip_int}} \}$ denote the corresponding \emph{mixed-integer hull}.
A \emph{cutting plane} for~\eqref{model_mip} is an inequality $\transpose{a}x \leq \beta$ that is valid for $P$ (and possibly invalid for some point computed during branch-and-cut).
Such a valid inequality \emph{induces the face} $F \coloneqq \{ x \in P : \transpose{a}x = \beta \}$ of $P$.
For more background on polyhedra we refer to~\cite{Schrijver86}.

To tackle specific problems, one often tries to find \emph{facet-defining} inequalities, which are inequalities whose induced face $F$ satisfies $\dim(F) = \dim(P) - 1$.
This is justified by the fact that any system $Cx \leq d$ with $P = \{x : Cx \leq d\}$ has to contain an inequality that induces $F$.
Since the dimension of a face can vary between $-1$ (no $x \in P$ satisfies $\transpose{a}x = \beta$) and $\dim(P)$ (all $x \in P$ satisfy $\transpose{a}x = \beta$), the following hypothesis is a reasonable generalization of facetness as a strength indicator.
\begin{hypothesis}
  \label{hypo_strength_dimension}
  The practical strength of an inequality correlates with the dimension of its induced face.
\end{hypothesis}
There is no unified notion of the practical strength of an inequality, but we will later define one that is related to its impact in a branch-and-bound algorithm.
The main goal of this paper is to computationally test this hypothesis for general-purpose cutting planes used in MIP solvers.

\paragraph{Outline.}
In Section~\ref{sec_dimension} we present the algorithm we used to compute dimensions.
Section~\ref{sec_strength} is dedicated to the score we used to assess a cutting plane's impact, and in Section~\ref{sec_computations} we present our findings, in particular regarding Hypothesis~\ref{hypo_strength_dimension}.

\section{Computing the dimension of a face}
\label{sec_dimension}

This section is concerned about how to effectively compute the dimension of $P$ or of one of its faces $F$ induced by some valid inequality $\transpose{a}x \leq \beta$.
In our experiments we will consider instances with several hundreds variables, and hence the enumeration of all vertices of $P$ or $F$ is typically impossible.
Instead, we use an oracle-based approach.
An \emph{optimization oracle} for $P$ is a black-box subroutine that can solve any linear program over the polyhedron $P$, i.e., for any given $w \in \R^n$ it can solve
\begin{equation}
  \text{max } \transpose{w}x \text{ s.t. } x \in P. \label{model_oracle}
\end{equation}
In case~\eqref{model_oracle} is feasible and bounded, the oracle shall return an optimal solution, and in case it is feasible and unbounded, it shall return an unbounded direction.
Note that for such an oracle we neither require all vertices of $P$ nor all valid (irredundant) inequalities.
Indeed, we can use any MIP solver and apply it to problem~\eqref{model_mip} with $c \coloneqq w$.

We now discuss how one can compute $\dim(P)$ by only accessing an optimization oracle.
To keep the presentation simple we assume that $P$ is bounded, although our implementation can handle unbounded polyhedra as well.
The algorithm maintains a set $X \subseteq P$ of affinely independent points and a system $Dx = e$ of valid equations, where $D$ has full row-rank $r$.
Hence, $\dim(X) \leq \dim(P) \leq n - r$ holds throughout and the algorithm works by either increasing $\dim(X)$ or $r$ in every iteration.
The details are provided in Algorithm~\ref{algo_affine_hull}.

\bigskip

\begin{algorithm}[H]
\TitleOfAlgo{Affine hull of a polytope via optimization oracle}
\label{algo_affine_hull}
\LinesNumbered
\setcounter{AlgoLine}{0}
\AlgoInput{Optimization oracle for a polytope $\emptyset \neq P \subseteq \R^n$.}
\AlgoOutput{%
  Affine basis $X$ of $\aff(P)$,
  non-redundant system $Dx = e$
  with $\aff(P) = \{ x \in \R^n : Dx = e\}$.
}
\BlankLine
Initialize $X \coloneqq \emptyset$ and $Dx = e$ empty.\;
\While{$|X| + 1 < n -|\rows(D)|$}
{%
  Compute a direction vector $d \coloneqq \aff(X)^\perp \setminus \lin(\rows(D))$. \label{algo_affine_hull_direction} \;
  Query the oracle to maximize $\transpose{d}x$ over $x \in P$ and let $x^+ \coloneqq \arg\max\{ \transpose{d}x : x \in P \}$. \;
  Query the oracle to maximize $-\transpose{d}x$ over $x \in P$ and let $x^- \coloneqq \arg\min\{ \transpose{d}x : x \in P \}$. \;
  \If{$\transpose{d}x^+ = \transpose{d}x^-$}
  {
    Add equation $\transpose{d}x = \transpose{d}x^+$ to system $Dx = e$. \label{algo_affine_hull_equation} \;
    \lIf{$X = \emptyset$}{set $X \coloneqq \{x^+\}$.} \label{algo_affine_hull_equation_point}
  }
  \lElseIf{$X = \emptyset$}{set $X \coloneqq \{ x^+, x^- \}$.} \label{algo_affine_hull_points}
  \Else
  {
    Let $\gamma \coloneqq \transpose{d}x$ for some $x \in X$. \;
    \lIf{$\transpose{d}x^+ \neq \gamma$}{set $X \coloneqq X \cup \{x^+\}$.} \label{algo_affine_hull_maximizer}
    \lElse{set $X \coloneqq X \cup \{x^-\}$.} \label{algo_affine_hull_minimizer}
  }
}
\Return{$(X, Dx = e)$.}
\end{algorithm}

\begin{proposition}
  For an optimization oracle for a non-empty polytope $P \subseteq \R^n$,
  Algorithm~\ref{algo_affine_hull} requires $2n$ oracle queries to compute a set $X \subseteq P$ with $|X| = \dim(P) + 1$ and a system $Dx = e$ of $n - \dim(P)$ equations satisfying
  \begin{equation*}
    \aff(P) = \aff(X) = \{ x \in \R^n : Dx = e \}.
  \end{equation*}
\end{proposition}

\begin{proof}
  Since every point $x \in X$ was computed by the optimization oracle, we have $X \subseteq P$.
  Moreover, $Dx = e$ is valid for $P$ since for each equation $\transpose{d}x = \gamma$, $\min \{ \transpose{d}x : x \in P \} = \gamma = \max \{ \transpose{d}x : x \in P \}$ holds.
  We now prove by induction on the number of iterations that in every iteration of the algorithm, the points $x \in X$ are affinely independent and that the rows of $D$ are linearly independent.

  Initially, this is invariant holds because $X = \emptyset$ and $D$ has no rows.
  Whenever an equation $\transpose{d}x = \transpose{d}x^+$ is added to $Dx = e$ in step~\ref{algo_affine_hull_equation}, the vector $d$ is linearly independent to the rows of $D$ (see step~\ref{algo_affine_hull_direction}).
  If in step~\ref{algo_affine_hull_equation_point}, $X$ is initialized with a single point, it is clearly affinely independent.
  Similarly, the two points in step~\ref{algo_affine_hull_points} also form an affinely independent set since $\transpose{d}x^+ > \transpose{d}x^-$ implies $x^+ \neq x^-$.
  Suppose $X$ is augmented by $\bar{x} \coloneqq x^+$ in step~\ref{algo_affine_hull_maximizer} or by $\bar{x} \coloneqq x^-$ in step~\ref{algo_affine_hull_minimizer}.
  Note that by the choice of $d$ in step~\ref{algo_affine_hull_direction}, $\transpose{d}x = \gamma$ holds for all $x \in X$.
  Due to $\transpose{d}\bar{x} \neq \gamma$, $X \cup \{\bar{x}\}$ remains affinely independent.

  After the first iteration, we have $|X| + |\rows(C)| = 2$, and in every further iteration, this quantity increases by $1$.
  Hence, the algorithm requires $n$ iterations, each of which performs $2$ oracle queries.
  The result follows.
  \qed
\end{proof}

For a proof that the number $2n$ of oracle queries is optimal we refer to~\cite{Walter16}.

\paragraph{Implementation details.}
We now describe some details of the implementation within the software framework \IPO{}~\cite{IPO}.
In the description of Algorithm~\ref{algo_affine_hull} we did not specify step~\ref{algo_affine_hull_direction} precisely.
While one may enforce the requirement $d \notin \lin(\rows(D))$ via $d \in \rows(D)^\perp$, it turned out that this is numerically less stable than first computing a basis of $\aff(X)^\perp$ and selecting a basis element $d$ that is not in the span of $D$'s rows.
Moreover, we can take $d$'s sparsity and other numerical properties into account.
Sparsity can speed-up the overall computation since very sparse objective vectors are sometimes easier to optimize for a MIP solver.
In theory, for all $x \in X$, the products $\transpose{d}x$ have the same value.
However, due to floating-point arithmetic the computed values may differ.
It turned out that preferring directions $d$ for which the range of these products is small helps to avoid numerical difficulties.

In our application we compute the dimension of $P$ and of several of its faces.
We can exploit this by caching all points $x \in P$ returned by an optimization oracle in a set $\bar{X}$.
Then, for each face $F$ induced by an inequality $\transpose{a}x \leq \beta$ we can then compute the set $\bar{X}_F \coloneqq \{ x \in \bar{X} : \transpose{a}x = \beta \}$. 
Before querying the oracle we can then check the set $\bar{X}_F$ for a point with sufficiently large objective value, which saves two calls to the MIP solvers.

Let $\bar{D}x = \bar{e}$ be the equation system returned by Algorithm~\ref{algo_affine_hull} for $P$.
Now, for a face $F$ induced by inequality $\transpose{a}x \leq \beta$ we can initialize $Dx = e$ in the algorithm by $\bar{D}x = \bar{e}$.
Moreover, if $\transpose{a}x = \beta$ is not implied by $Dx = e$ we add this equation to $Dx = e$ as well.

Since the algorithm is implemented in floating-point arithmetic, errors can occur which may lead to wrong dimension results.
We checked the results using an exact arithmetic implementation of Algorithm~\ref{algo_affine_hull} for easier instances (with less than 200 variables).
For these, the computed dimension varied by at most $2$ from the true dimension.

In a first implementation, our code frequently reported dimension $-1$, and it turned out that often the right-hand side was only slightly larger than needed to make the inequality supporting.
Thus, for each inequality $\transpose{a}x \leq \beta$, normalized to $||a||_2 = 1$, we computed $\beta^{\text{true}} \coloneqq \max \{ \transpose{a}x : x \in P \}$ by a single oracle query.
Whenever we observed $\beta < \beta^{\text{true}} - 10^{-4}$, we considered the cut $\transpose{a}x \leq \beta$ as invalid (indicated by the symbol \faBolt{}).
If $\beta > \beta^{\text{true}} + 10^{-4}$, we declare the cut to be non-supporting.
Otherwise, we replace $\beta$ by $\beta^{\text{true}}$ before running Algorithm~\ref{algo_affine_hull}.

\section{Measuring the strength of a single inequality}
\label{sec_strength}

In this section we introduce our \emph{cut impact measure} for indicating, for a given cutting plane $\transpose{a}x \leq \beta$, how useful its addition in the context of branch-and-cut is.
The main goal of solving an LP at a branch-and-bound node is to determine a dual bound of the current subproblem.
If this bound is at most the value of the best feasible solution known so far, then the subproblem can be discarded.
Thus, we consider the value of such a bound (after adding a certain cut) in relation to the problem's optimum $z^\star \coloneqq \max \{ \transpose{c}x : Ax \leq b,~ x_i \in \Z ~\forall i \in I \}$ and the value $z^{\mathrm{LP}} \coloneqq \max \{ \transpose{c}x : Ax \leq b \}$ of the LP relaxation.

Our first approach was to just evaluate the dual bound obtained from the LP relaxation $Ax \leq b$ augmented by $\transpose{a}x \leq \beta$.
However, adding a single inequality often does not cut off the optimal face of the LP relaxation, which means that the bound does not change.
In a second attempt we tried to evaluate the dual bound of the LP relaxation augmented by a random selection of $k$ cutting planes.
However, the variance of the resulting cut impact measure was very large even after averaging over 10.000 such selections.

As a consequence, we discarded cut impact measures based on the combined impact of several cutting planes.
Instead we carried out the following steps for given cutting planes $\transpose{a_1} x \leq \beta_1$, $\transpose{a_2} x \leq \beta_2$, \dots, $\transpose{a_k} x \leq \beta_k$:
\begin{enumerate}
\item
  Compute the optimum $z^\star$ and optimum solution $x^\star \in P$.
\item
  Compute $z^{\mathrm{LP}}$.
\item
  For $i=0,1,2,\dotsc,k$, solve
  \begin{equation*}
    \text{max } \transpose{c}x \text{ s.t. } Ax \leq b,~ x_j \in \Z ~\forall j \in I \text{ and } \transpose{a_i} x \leq \beta_i \text{ if } i \geq 1
  \end{equation*}
  with $x^\star$ as an initial incumbent, with presolve, cutting planes and heuristics disabled and where a time limit of $60$ seconds is enforced.
\item
  Let $N$ denote the minimum number of branch-and-bound nodes used in any these $k+1$ runs.
\item
  For $i=0,1,\dotsc,k$, let $z^i$ be the dual bound obtained in run $i$ when stopping after $N$ branch-and-bound nodes.
  For $i \geq 1$, the \emph{closed gap of cut $i$} is defined as $(z^i - z^{\mathrm{LP}}) / (z^\star - z^{\mathrm{LP}})$.
  For $i = 0$, this yields the \emph{closed gap without cuts}.
\end{enumerate}
\paragraph{Remarks.}
The closed gap is essentially the dual bound, normalized such that a value of $0$ means no bound improvement over the root LP bound without cuts and a value of $1$ means that the instance was solved to optimality.
The effective limit of $N$ branch-and-bound nodes was introduced such that all runs reach this limit.
This circumvents the question of how to compare runs in which the problem was solved to optimality with those that could not solve it.
For all runs, presolve and domain propagation were disabled due to our focus on the branch-and-bound algorithm itself.
To avoid interaction with heuristics, the latter were disabled, but the optimal solution $x^\star$ was provided.

\section{Computational study}
\label{sec_computations}

In order to test Hypothesis~\ref{hypo_strength_dimension} we considered the 65 instances from the MIPLIB~3~\cite{MIPLIB3}.
For each of them we computed the dimension of the mixed-integer hull $P$, imposing a time limit of 10~minutes\footnote{All experiments were carried out on a single core of an Intel Core~i3 CPU running at 2.10\,GHz with 8\ GB RAM.}.
Moreover, we ran the state-of-the-art solver \SCIP{}~\cite{SCIP7} and collected all cutting planes generated in the root node, including those that were discarded by \SCIP{}'s cut selection routine\footnote{We disabled presolve, domain propagation, dual reductions, symmetry, and restarts.}.
For each of the cuts we computed the dimension of its induced face (see Section~\ref{sec_dimension}) as well as the closed gap after processing $N$ (as defined in Section~\ref{sec_strength}) branch-and-bound nodes.

While for instances \texttt{air03}, \texttt{air05}, \texttt{nw04}, no cuts were generated by \SCIP{}, our implementation of Algorithm~\ref{algo_affine_hull} ran into numerical difficulties during the computation of $\dim(P)$ for \texttt{set1ch}.
Moreover, $\dim(P)$ could not be computed within 10~minutes for instances \texttt{air04}, \texttt{arki001}, \texttt{cap6000}, \texttt{dano3mip}, \texttt{danoint}, \texttt{dsbmip}, \texttt{fast0507}, \texttt{gesa2\_0}, \texttt{gesa3\_0}, \texttt{l152lav}, \texttt{mas74}, \texttt{misc06}, \texttt{mitre}, \texttt{mkc}, \texttt{mod010}, \texttt{mod011}, \texttt{pk1}, \texttt{pp08aCUTS}, \texttt{pp08a}, \texttt{qnet1}, \texttt{rentacar}, \texttt{rout} and \texttt{swath}.

We evaluated the remaining 37 instances, whose characteristic data is shown in Table~\ref{table_instances}.
We distinguish how many cuts \SCIP{} found by which cut separation method, in particular to investigate whether certain separation routines tend to generate cuts with low- or high-dimensional faces.

We verified some of the invalid cutting planes manually, i.e., checked that these cuts are generated by \SCIP{} and that there exists a feasible solution that is indeed cut off.
The most likely reason for their occurrence is that \SCIP{} performed dual reductions although we disabled them via corresponding parameters\footnote{We set \texttt{misc/allowweakdualreds} and \texttt{misc/allowstrongdualreds} to \texttt{false}.}.

\DeclareDocumentCommand\sepaLinear{}{\begin{tikzpicture}
    \filldraw[draw=black,fill=black!20!white] (0,0) circle [radius=4pt];
    \filldraw[ultra thin,draw=black,fill=black] (0,0) -- +(90-30:4pt) arc[start angle=90-30, delta angle=60, radius=4pt] -- cycle;
  \end{tikzpicture}}
\DeclareDocumentCommand\sepaMIR{}{\begin{tikzpicture}
    \filldraw[draw=black,fill=black!20!white] (0,0) circle [radius=4pt];
    \filldraw[ultra thin,draw=black,fill=black] (0,0) -- +(90-3*30:4pt) arc[start angle=90-3*30, delta angle=60, radius=4pt] -- cycle;
  \end{tikzpicture}}
\DeclareDocumentCommand\sepaZeroHalf{}{\begin{tikzpicture}
    \filldraw[draw=black,fill=black!20!white] (0,0) circle [radius=4pt];
    \filldraw[ultra thin,draw=black,fill=black] (0,0) -- +(90-5*30:4pt) arc[start angle=90-5*30, delta angle=60, radius=4pt] -- cycle;
  \end{tikzpicture}}
\DeclareDocumentCommand\sepaStrongCg{}{\begin{tikzpicture}
    \filldraw[draw=black,fill=black!20!white] (0,0) circle [radius=4pt];
    \filldraw[ultra thin,draw=black,fill=black] (0,0) -- +(90-7*30:4pt) arc[start angle=90-7*30, delta angle=60, radius=4pt] -- cycle;
  \end{tikzpicture}}
\DeclareDocumentCommand\sepaFlowcover{}{\begin{tikzpicture}
    \filldraw[draw=black,fill=black!20!white] (0,0) circle [radius=4pt];
    \filldraw[ultra thin,draw=black,fill=black] (0,0) -- +(90-9*30:4pt) arc[start angle=90-9*30, delta angle=60, radius=4pt] -- cycle;
  \end{tikzpicture}}
\DeclareDocumentCommand\sepaMcf{}{\begin{tikzpicture}
    \filldraw[draw=black,fill=black!20!white] (0,0) circle [radius=4pt];
    \filldraw[ultra thin,draw=black,fill=black] (0,0) -- +(90-11*30:4pt) arc[start angle=90-11*30, delta angle=60, radius=4pt] -- cycle;
  \end{tikzpicture}}
  
\begin{table}[htb!]
  \caption{%  
    Characteristics of the relevant 37 instances with number of successfully analyzed cuts, failures (numerical problems $0/0$, timeouts \faClockO{}, invalid cuts \faBolt{}), dimension of the mixed-integer hull, and number $N$ of branch-and-bound nodes (see Section~\ref{sec_strength}). \newline
    \sepaLinear{} -- Lifted extended weight inequalities~\cite{Wolter06,Martin99,Weismantel97} \newline % L
    \sepaMIR{} -- Complemented Mixed-Integer Rounding (c-MIR) inequalities~\cite{MarchandW01,Wolter06} \newline % R
    \sepaZeroHalf{} -- $\{0,1/2\}$-Chv{\'a}tal-Gomory inequalities~\cite{KosterZK09,CapraraF96} \newline % Z
    \sepaStrongCg{} -- Strengthened Chv{\'a}tal-Gomory inequalities~\cite{LetchfordL02b}  \newline % C
    \sepaFlowcover{} -- Lifted flow-cover inequalities~\cite{GuNS99} \newline % F
    \sepaMcf{} -- Multi-commodity flow inequalities~\cite{AchterbergR10} % M
    }
  \label{table_instances}
  \begin{center}
    {\small{%
    \setlength{\tabcolsep}{1ex}
    \begin{tabular}{lr|rrrrrr|rrr|rrr}
      \textbf{Instance}
      & \textbf{Cuts}
      & \multicolumn{6}{c|}{\textbf{Analyzed by class}}
      & \multicolumn{3}{c|}{\textbf{Failed}}
      & \multicolumn{1}{c}{\textbf{Dim.}}
      & \multicolumn{1}{c}{\textbf{B\&B}}
      \\
      & \multicolumn{1}{c|}{\textbf{total}}
      & \multicolumn{1}{c|}{\sepaLinear{}}
      & \multicolumn{1}{c|}{\sepaMIR{}}
      & \multicolumn{1}{c|}{\sepaZeroHalf{}}
      & \multicolumn{1}{c|}{\sepaStrongCg{}}
      & \multicolumn{1}{c|}{\sepaFlowcover{}}
      & \multicolumn{1}{c|}{\sepaMcf{}}
      & \multicolumn{1}{c|}{$0/0$}
      & \multicolumn{1}{c|}{\faClockO{}}
      & \multicolumn{1}{c|}{\faBolt{}}
      & \multicolumn{1}{c}{\textbf{of \emph{P}}}
      & \multicolumn{1}{c}{\textbf{nodes}}
      \\ \hline
      bell3a & 25 &  & 25 &  &  &  &  &  &  &  & 121 & 53075 \\
      bell5 & 53 &  & 36 &  &  &  & 1 & 11 &  & 5 & 97 & 7815 \\
      blend2 & 8 &  & 7 &  &  &  &  &  &  & 1 & 245 & 597 \\
      dcmulti & 172 &  & 137 &  &  &  &  & 14 &  & 21 & 467 & 773 \\
      egout & 125 &  & 97 &  &  & 1 &  & 24 &  & 3 & 41 & 1 \\
      enigma & 82 &  & 59 & 1 & 17 & 5 &  &  &  &  & 3 & 1 \\
      fiber & 533 & 127 & 139 & 7 & 4 & 9 & 207 & 33 & 6 & 1 & 946 & 23395 \\
      fixnet6 & 772 &  & 612 &  &  &  & 83 & 42 &  & 35 & 779 & 53900 \\
      flugpl & 42 &  & 42 &  &  &  &  &  &  &  & 9 & 1753 \\
      gen & 28 & 17 & 8 &  &  &  &  & 3 &  &  & 540 & 21 \\
      gesa2 & 470 & 16 & 419 & 2 & 6 &  &  & 27 &  &  & 1176 & 21114 \\
      gesa3 & 259 &  & 194 & 2 & 6 & 15 &  & 38 & 4 &  & 1104 & 729 \\
      gt2 & 143 &  & 92 & 3 & 39 & 8 &  &  & 1 &  & 188 & 1 \\
      harp2 & 1028 & 659 & 210 & 9 & 4 & 112 &  & 5 & 28 & 1 & 1300 & 13748 \\
      khb05250 & 122 &  & 78 &  &  &  &  & 5 &  & 39 & 1229 & 425 \\
      lseu & 125 & 35 & 85 & 2 & 3 &  &  &  &  &  & 89 & 3495 \\
      markshare1 & 107 &  & 104 &  &  & 3 &  &  &  &  & 50 & 318260 \\
      markshare2 & 84 &  & 79 &  &  & 3 &  &  & 2 &  & 60 & 286347 \\
      mas76 & 225 &  & 204 &  &  &  &  &  & 1 & 20 & 151 & 211381 \\
      misc03 & 634 & 3 & 283 & 33 & 311 &  &  &  & 4 &  & 116 & 13 \\
      misc07 & 808 &  & 286 & 34 & 440 &  &  & 36 & 12 &  & 204 & 22839 \\
      mod008 & 441 & 119 & 272 &  &  &  &  &  & 3 & 47 & 319 & 1158 \\
      modglob & 268 &  & 186 &  &  &  &  & 3 & 1 & 78 & 327 & 112516 \\
      noswot & 164 &  & 151 &  & 11 &  &  & 2 &  &  & 120 & 160344 \\
      p0033 & 94 & 14 & 40 &  & 10 &  &  & 30 &  &  & 27 & 127 \\
      p0201 & 263 & 12 & 152 & 14 & 78 &  &  & 7 &  &  & 139 & 21 \\
      p0282 & 904 & 368 & 441 & 6 & 15 & 1 &  & 4 &  & 69 & 282 & 57 \\
      p0548 & 577 & 159 & 213 & 10 & 16 & 62 &  & 117 &  &  & 520 & 921 \\
      p2756 & 1000 & 82 & 96 & 22 & 48 & 55 &  & 26 & 671 &  & 2716 & 15149 \\
      qiu & 63 &  & 51 &  &  &  &  & 2 & 10 &  & 709 & 10406 \\
      qnet1 & 162 &  & 95 & 10 &  & 47 &  & 6 & 4 &  & 1233 & 5 \\
      rgn & 278 &  & 168 &  &  & 65 &  & 45 &  &  & 160 & 1691 \\
      seymour & 6246 &  & 656 & 53 & 5528 &  &  & 9 & 0 &  & 1255 & 9 \\
      stein27 & 886 &  & 517 & 9 & 360 &  &  &  &  &  & 27 & 3673 \\
      stein45 & 1613 &  & 1221 & 10 & 382 &  &  &  &  &  & 45 & 45371 \\
      vpm1 & 281 &  & 129 &  &  & 32 &  & 117 &  & 3 & 288 & 27881 \\
      vpm2 & 353 &  & 251 & 1 &  & 30 &  & 63 &  & 8 & 286 & 172271
    \end{tabular}
    }}
  \end{center}
\end{table}

\clearpage

Since the results on face dimensions and cut strength turned out to be very instance-specific, we created one plot per instance.
We omit the ones for \texttt{p2756} (too many failures, see Table~\ref{table_instances}), \texttt{blend2} (only 7 cuts analyzed), \texttt{enigma} ($\dim(P) = 3$ is very small), and for \texttt{markshare1}, \texttt{markshare2} and \texttt{noswot} (all cuts were ineffective).
Moreover, we present several plots in the Appendix~\ref{appendix_plots} since these are similar to those of other instances.

The plots show the dimension of the cuts (horizontal axis, rounded to 19 groups) together with their closed gap (vertical axis, 14 groups) according to Section~\ref{sec_strength}.
Each circle corresponds to a nonempty set of cuts, where the segments depict the respective cut classes (see Table~\ref{table_instances}) and their color depicts the number $k$ of cuts, where the largest occurring number $L$ is specified in the caption.
The colors are red ($0.9L < k \leq L$), orange ($0.7L < k \leq 0.9L$), yellow ($0.5L < k \leq 0.7L$), green ($0.3L < k \leq 0.5L$), turquoise ($0.1L < k \leq 0.3L$) and blue ($1 \leq k < 0.1L$).
For instance, the circle for \texttt{fixnet6} containing a red and a turquoise segment subsumes cutting planes with face dimensions between $730$ and $777$, and closed gap of approximately $0.45$.
As the legend next to the plot indicates, this circle represents $k \in [109,121]$ c-MIR cuts and $k' \in [13,36]$ multi-commodity flow cuts.
The dashed horizontal line indicates the closed gap without cuts (see Section~\ref{sec_strength}).

\bigskip

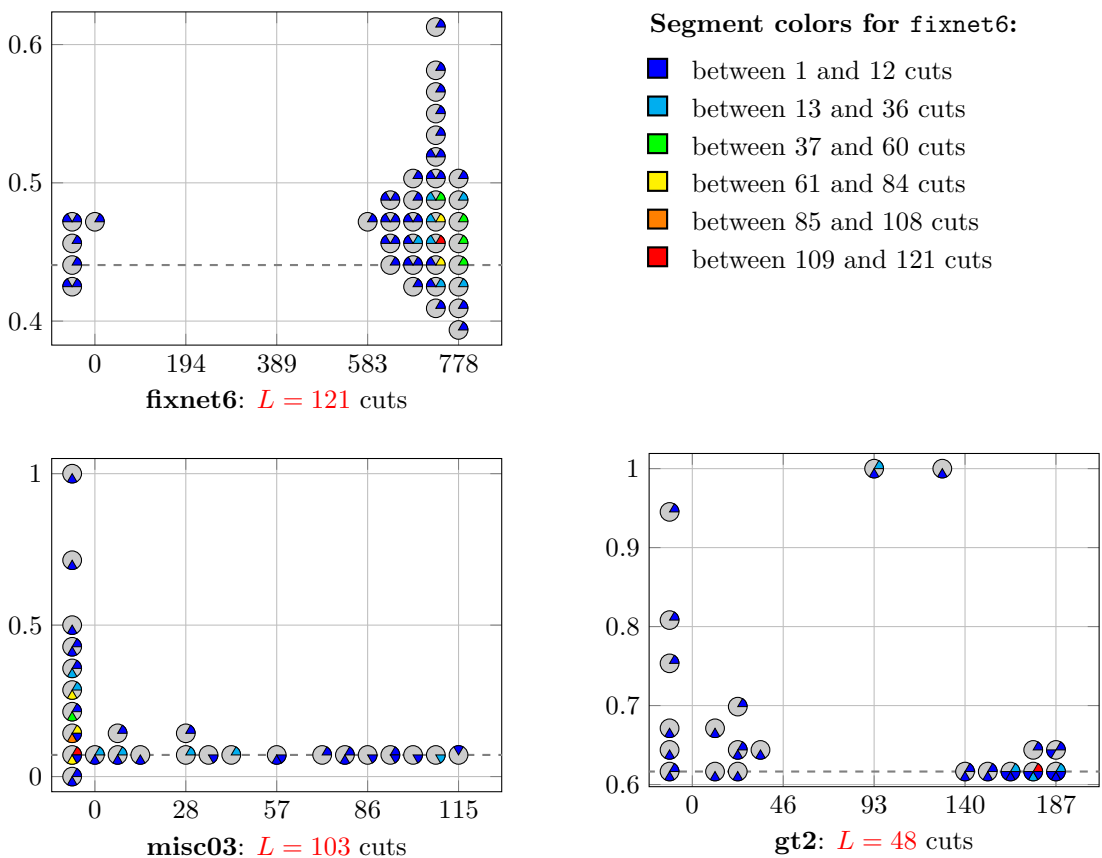
\DeclareDocumentCommand\drawMark{mm}{\filldraw[draw=black,fill=black!20!white] (axis cs:#1,#2) circle [radius=3.5pt];}
\DeclareDocumentCommand\drawTriangle{mmmm}{\filldraw[ultra thin,draw=black,fill=#4] (axis cs:#1,#2) -- +(#3:3.5pt) arc[start angle=#3, delta angle=60, radius=3.5pt] -- cycle;}

% fixnet6: high dimensions best
% misc03: empty face best.

% harp2: cuts don't help
% mod008: wide spread

% qiu: all cuts low dimension.
% qnet: all cuts high dimension.

% seymour: viel leer trotz vieler Cuts
% misc07: cut classes yield different dimensions.

% fiber: majority of cuts concentrated near nocut -> maybe quality of a single cut is not so meaningful. Maybe dim(P) becomes more important when combining many cuts.
\message{Instance <fixnet6>.}
% \message{Instance <fixnet6> has 695 / 772 cuts.}
% \message{Instance <fixnet6> has main dimension 779.}
% \message{Instance <fixnet6> required 53900 nodes.}
% \message{Instance <fixnet6> has 612 cuts of type 'R'.}
% \message{Instance <fixnet6> has 83 cuts of type 'M'.}
  \begin{tikzpicture}
    \begin{axis}[
      enlargelimits = 0.05,
      xlabel = {\textbf{fixnet6}: \textcolor{red}{$L = 121$} cuts},
      width = 75mm, % old: 66
      height = 60mm, % old: 51
      xtick={0,4,8,12,16},
      xticklabels={0,$194$,$389$,$583$,$778$},
      xmin = -1,
      xmax = 17,
      ymin = 0.393586,
      ymax = 0.612671,
      grid = major,
      colormap = {blueorangered}{rgb255(0cm)=(0,0,255); rgb255(1cm)=(255,128,0); rgb255(2cm)=(255,0,0)},
      colormap name = blueorangered,
    ]
    \draw[black!50!white,thick,dashed] (axis cs:-2,0.440533) -- (axis cs:18,0.440533);
    \drawMark{-1}{0.424884} \drawTriangle{-1}{0.424884}{-240}{blue} \drawTriangle{-1}{0.424884}{0}{blue}
    \drawMark{-1}{0.440533} \drawTriangle{-1}{0.440533}{0}{blue}
    \drawMark{-1}{0.456182} \drawTriangle{-1}{0.456182}{0}{blue}
    \drawMark{-1}{0.471831} \drawTriangle{-1}{0.471831}{-240}{blue} \drawTriangle{-1}{0.471831}{0}{blue}
    \drawMark{0}{0.471831} \drawTriangle{0}{0.471831}{0}{blue}
    \drawMark{12}{0.471831} \drawTriangle{12}{0.471831}{0}{blue}
    \drawMark{13}{0.440533} \drawTriangle{13}{0.440533}{0}{blue}
    \drawMark{13}{0.456182} \drawTriangle{13}{0.456182}{-240}{blue} \drawTriangle{13}{0.456182}{0}{blue}
    \drawMark{13}{0.471831} \drawTriangle{13}{0.471831}{-240}{blue} \drawTriangle{13}{0.471831}{0}{blue}
    \drawMark{13}{0.48748} \drawTriangle{13}{0.48748}{-240}{blue} \drawTriangle{13}{0.48748}{0}{blue}
    \drawMark{14}{0.424884} \drawTriangle{14}{0.424884}{0}{blue}
    \drawMark{14}{0.440533} \drawTriangle{14}{0.440533}{-240}{blue} \drawTriangle{14}{0.440533}{0}{blue}
    \drawMark{14}{0.456182} \drawTriangle{14}{0.456182}{-240}{blue} \drawTriangle{14}{0.456182}{0}{cyan}
    \drawMark{14}{0.471831} \drawTriangle{14}{0.471831}{-240}{blue} \drawTriangle{14}{0.471831}{0}{blue}
    \drawMark{14}{0.48748} \drawTriangle{14}{0.48748}{0}{blue}
    \drawMark{14}{0.503128} \drawTriangle{14}{0.503128}{0}{blue}
    \drawMark{15}{0.409235} \drawTriangle{15}{0.409235}{0}{blue}
    \drawMark{15}{0.424884} \drawTriangle{15}{0.424884}{-240}{blue} \drawTriangle{15}{0.424884}{0}{cyan}
    \drawMark{15}{0.440533} \drawTriangle{15}{0.440533}{-240}{blue} \drawTriangle{15}{0.440533}{0}{yellow}
    \drawMark{15}{0.456182} \drawTriangle{15}{0.456182}{-240}{cyan} \drawTriangle{15}{0.456182}{0}{red}
    \drawMark{15}{0.471831} \drawTriangle{15}{0.471831}{-240}{cyan} \drawTriangle{15}{0.471831}{0}{yellow}
    \drawMark{15}{0.48748} \drawTriangle{15}{0.48748}{-240}{cyan} \drawTriangle{15}{0.48748}{0}{green}
    \drawMark{15}{0.503128} \drawTriangle{15}{0.503128}{-240}{blue} \drawTriangle{15}{0.503128}{0}{blue}
    \drawMark{15}{0.518777} \drawTriangle{15}{0.518777}{-240}{blue} \drawTriangle{15}{0.518777}{0}{blue}
    \drawMark{15}{0.534426} \drawTriangle{15}{0.534426}{0}{blue}
    \drawMark{15}{0.550075} \drawTriangle{15}{0.550075}{0}{blue}
    \drawMark{15}{0.565724} \drawTriangle{15}{0.565724}{0}{blue}
    \drawMark{15}{0.581373} \drawTriangle{15}{0.581373}{0}{blue}
    \drawMark{15}{0.612671} \drawTriangle{15}{0.612671}{0}{blue}
    \drawMark{16}{0.393586} \drawTriangle{16}{0.393586}{0}{blue}
    \drawMark{16}{0.409235} \drawTriangle{16}{0.409235}{0}{blue}
    \drawMark{16}{0.424884} \drawTriangle{16}{0.424884}{0}{cyan}
    \drawMark{16}{0.440533} \drawTriangle{16}{0.440533}{0}{green}
    \drawMark{16}{0.456182} \drawTriangle{16}{0.456182}{0}{green}
    \drawMark{16}{0.471831} \drawTriangle{16}{0.471831}{0}{green}
    \drawMark{16}{0.48748} \drawTriangle{16}{0.48748}{0}{cyan}
    \drawMark{16}{0.503128} \drawTriangle{16}{0.503128}{0}{blue}
    \end{axis}
  \end{tikzpicture}
  \begin{tikzpicture}[overlay, xshift=18mm, yshift=6mm]
    \node[anchor=west] at (-0.1,4.6) {\textbf{Segment colors for \texttt{fixnet6}:}};
    \foreach \y/\c/\a/\b in {0/blue/1/12, 1/cyan/13/36, 2/green/37/60, 3/yellow/61/84, 4/orange/85/108, 5/red/109/121}
    {
      \filldraw[thick,draw=black,fill=\c] (0,3.9-0.5*\y) rectangle +(0.25,0.25);
      \node[anchor=west] at (0.45,4-0.5*\y) {between $\a$ and $\b$ cuts};
    }
  \end{tikzpicture}

\bigskip

\message{Instance <misc03>.}
% \message{Instance <misc03> has 630 / 634 cuts.}
% \message{Instance <misc03> has main dimension 116.}
% \message{Instance <misc03> required 13 nodes.}
% \message{Instance <misc03> has 3 cuts of type 'L'.}
% \message{Instance <misc03> has 283 cuts of type 'R'.}
% \message{Instance <misc03> has 33 cuts of type 'Z'.}
% \message{Instance <misc03> has 311 cuts of type 'C'.}
  \begin{tikzpicture}
    \begin{axis}[
      enlargelimits = 0.05,
      xlabel = {\textbf{misc03}: \textcolor{red}{$L = 103$} cuts},
      width = 75mm, % old: 66
      height = 60mm, % old: 51
      xtick={0,4,8,12,16},
      xticklabels={0,$28$,$57$,$86$,$115$},
      xmin = -1,
      xmax = 17,
      ymin = 0,
      ymax = 1,
      grid = major,
      colormap = {blueorangered}{rgb255(0cm)=(0,0,255); rgb255(1cm)=(255,128,0); rgb255(2cm)=(255,0,0)},
      colormap name = blueorangered,
    ]
    \draw[black!50!white,thick,dashed] (axis cs:-2,0.0714286) -- (axis cs:18,0.0714286);
    \drawMark{-1}{0} \drawTriangle{-1}{0}{-120}{blue} \drawTriangle{-1}{0}{0}{blue}
    \drawMark{-1}{0.0714286} \drawTriangle{-1}{0.0714286}{-120}{yellow} \drawTriangle{-1}{0.0714286}{0}{red} \drawTriangle{-1}{0.0714286}{-60}{blue}
    \drawMark{-1}{0.142857} \drawTriangle{-1}{0.142857}{-120}{orange} \drawTriangle{-1}{0.142857}{0}{yellow} \drawTriangle{-1}{0.142857}{-60}{blue}
    \drawMark{-1}{0.214286} \drawTriangle{-1}{0.214286}{-120}{green} \drawTriangle{-1}{0.214286}{0}{blue}
    \drawMark{-1}{0.285714} \drawTriangle{-1}{0.285714}{-120}{yellow} \drawTriangle{-1}{0.285714}{0}{cyan}
    \drawMark{-1}{0.357143} \drawTriangle{-1}{0.357143}{-120}{cyan} \drawTriangle{-1}{0.357143}{0}{blue}
    \drawMark{-1}{0.428571} \drawTriangle{-1}{0.428571}{-120}{blue} \drawTriangle{-1}{0.428571}{0}{blue}
    \drawMark{-1}{0.5} \drawTriangle{-1}{0.5}{-120}{blue}
    \drawMark{-1}{0.714286} \drawTriangle{-1}{0.714286}{-120}{blue}
    \drawMark{-1}{1} \drawTriangle{-1}{1}{-120}{blue}
    \drawMark{0}{0.0714286} \drawTriangle{0}{0.0714286}{-120}{blue} \drawTriangle{0}{0.0714286}{0}{cyan}
    \drawMark{1}{0.0714286} \drawTriangle{1}{0.0714286}{-120}{blue} \drawTriangle{1}{0.0714286}{0}{cyan}
    \drawMark{1}{0.142857} \drawTriangle{1}{0.142857}{0}{blue}
    \drawMark{2}{0.0714286} \drawTriangle{2}{0.0714286}{-120}{blue}
    \drawMark{4}{0.0714286} \drawTriangle{4}{0.0714286}{0}{cyan}
    \drawMark{4}{0.142857} \drawTriangle{4}{0.142857}{0}{blue}
    \drawMark{5}{0.0714286} \drawTriangle{5}{0.0714286}{-60}{blue}
    \drawMark{6}{0.0714286} \drawTriangle{6}{0.0714286}{0}{cyan}
    \drawMark{8}{0.0714286} \drawTriangle{8}{0.0714286}{-120}{blue} \drawTriangle{8}{0.0714286}{-60}{blue}
    \drawMark{10}{0.0714286} \drawTriangle{10}{0.0714286}{0}{blue}
    \drawMark{11}{0.0714286} \drawTriangle{11}{0.0714286}{-120}{blue} \drawTriangle{11}{0.0714286}{0}{blue}
    \drawMark{12}{0.0714286} \drawTriangle{12}{0.0714286}{-60}{blue}
    \drawMark{13}{0.0714286} \drawTriangle{13}{0.0714286}{0}{blue} \drawTriangle{13}{0.0714286}{-60}{blue}
    \drawMark{14}{0.0714286} \drawTriangle{14}{0.0714286}{-60}{blue}
    \drawMark{15}{0.0714286} \drawTriangle{15}{0.0714286}{-60}{cyan}
    \drawMark{16}{0.0714286} \drawTriangle{16}{0.0714286}{60}{blue}
    \end{axis}
  \end{tikzpicture}
\hfill
\message{Instance <gt2>.}
% \message{Instance <gt2> has 142 / 143 cuts.}
% \message{Instance <gt2> has main dimension 188.}
% \message{Instance <gt2> required 1 nodes.}
% \message{Instance <gt2> has 92 cuts of type 'R'.}
% \message{Instance <gt2> has 3 cuts of type 'Z'.}
% \message{Instance <gt2> has 39 cuts of type 'C'.}
% \message{Instance <gt2> has 8 cuts of type 'F'.}
  \begin{tikzpicture}
    \begin{axis}[
      enlargelimits = 0.05,
      xlabel = {\textbf{gt2}: \textcolor{red}{$L = 48$} cuts},
      width = 75mm, % old: 66
      height = 60mm, % old: 51
      xtick={0,4,8,12,16},
      xticklabels={0,$46$,$93$,$140$,$187$},
      xmin = -1,
      xmax = 17,
      ymin = 0.616471,
      ymax = 1,
      grid = major,
      colormap = {blueorangered}{rgb255(0cm)=(0,0,255); rgb255(1cm)=(255,128,0); rgb255(2cm)=(255,0,0)},
      colormap name = blueorangered,
    ]
    \draw[black!50!white,thick,dashed] (axis cs:-2,0.616471) -- (axis cs:18,0.616471);
    \drawMark{-1}{0.616471} \drawTriangle{-1}{0.616471}{-120}{blue} \drawTriangle{-1}{0.616471}{0}{blue}
    \drawMark{-1}{0.643866} \drawTriangle{-1}{0.643866}{-120}{blue}
    \drawMark{-1}{0.671261} \drawTriangle{-1}{0.671261}{-120}{blue}
    \drawMark{-1}{0.753446} \drawTriangle{-1}{0.753446}{0}{blue}
    \drawMark{-1}{0.808235} \drawTriangle{-1}{0.808235}{0}{blue}
    \drawMark{-1}{0.94521} \drawTriangle{-1}{0.94521}{0}{blue}
    \drawMark{1}{0.616471} \drawTriangle{1}{0.616471}{-120}{blue}
    \drawMark{1}{0.671261} \drawTriangle{1}{0.671261}{-120}{blue}
    \drawMark{2}{0.616471} \drawTriangle{2}{0.616471}{-120}{blue}
    \drawMark{2}{0.643866} \drawTriangle{2}{0.643866}{-120}{blue} \drawTriangle{2}{0.643866}{0}{blue}
    \drawMark{2}{0.698656} \drawTriangle{2}{0.698656}{0}{blue}
    \drawMark{3}{0.643866} \drawTriangle{3}{0.643866}{-120}{blue}
    \drawMark{8}{1} \drawTriangle{8}{1}{-120}{blue} \drawTriangle{8}{1}{0}{cyan}
    \drawMark{11}{1} \drawTriangle{11}{1}{-120}{blue}
    \drawMark{12}{0.616471} \drawTriangle{12}{0.616471}{-120}{blue} \drawTriangle{12}{0.616471}{0}{blue}
    \drawMark{13}{0.616471} \drawTriangle{13}{0.616471}{-120}{blue} \drawTriangle{13}{0.616471}{0}{blue}
    \drawMark{14}{0.616471} \drawTriangle{14}{0.616471}{-120}{blue} \drawTriangle{14}{0.616471}{-180}{blue} \drawTriangle{14}{0.616471}{0}{cyan} \drawTriangle{14}{0.616471}{-60}{blue}
    \drawMark{15}{0.616471} \drawTriangle{15}{0.616471}{-120}{cyan} \drawTriangle{15}{0.616471}{-180}{blue} \drawTriangle{15}{0.616471}{0}{red} \drawTriangle{15}{0.616471}{-60}{blue}
    \drawMark{15}{0.643866} \drawTriangle{15}{0.643866}{0}{blue}
    \drawMark{16}{0.616471} \drawTriangle{16}{0.616471}{-120}{blue} \drawTriangle{16}{0.616471}{-180}{blue} \drawTriangle{16}{0.616471}{0}{cyan} \drawTriangle{16}{0.616471}{-60}{blue}
    \drawMark{16}{0.643866} \drawTriangle{16}{0.643866}{-180}{blue} \drawTriangle{16}{0.643866}{0}{blue}
    \end{axis}
  \end{tikzpicture}

\bigskip
  
The first three plots already highlight that the results are very heterogeneous: while the faces of the strongest cuts in \texttt{fixnet6} have a high dimension, the strongest ones for \texttt{misc03} are not even supporting.
Even when considering non-supporting cuts as outliers, the dimension does not indicate practical strength, as the plot for \texttt{gt2} shows.
A quick look at the other plots lets us conclude that Hypothesis~\ref{hypo_strength_dimension} is false --- at least for the strength measure from Section~\ref{sec_strength}.

\bigskip

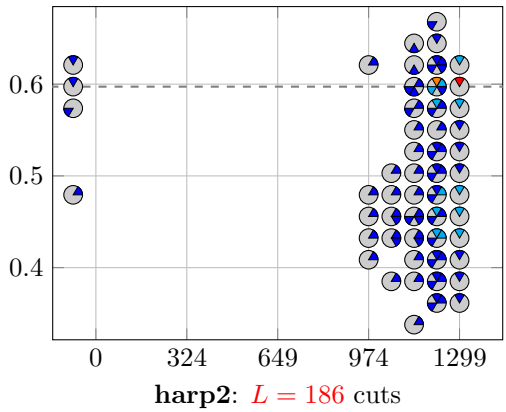
\message{Instance <harp2>.}
% \message{Instance <harp2> has 994 / 1028 cuts.}
% \message{Instance <harp2> has main dimension 1300.}
% \message{Instance <harp2> required 13748 nodes.}
% \message{Instance <harp2> has 659 cuts of type 'L'.}
% \message{Instance <harp2> has 210 cuts of type 'R'.}
% \message{Instance <harp2> has 9 cuts of type 'Z'.}
% \message{Instance <harp2> has 4 cuts of type 'C'.}
% \message{Instance <harp2> has 112 cuts of type 'F'.}
  \begin{tikzpicture}
    \begin{axis}[
      enlargelimits = 0.05,
      xlabel = {\textbf{harp2}: \textcolor{red}{$L = 186$} cuts},
      width = 75mm, % old: 66
      height = 60mm, % old: 51
      xtick={0,4,8,12,16},
      xticklabels={0,$324$,$649$,$974$,$1299$},
      xmin = -1,
      xmax = 17,
      ymin = 0.337878,
      ymax = 0.668018,
      grid = major,
      colormap = {blueorangered}{rgb255(0cm)=(0,0,255); rgb255(1cm)=(255,128,0); rgb255(2cm)=(255,0,0)},
      colormap name = blueorangered,
    ]
    \draw[black!50!white,thick,dashed] (axis cs:-2,0.597274) -- (axis cs:18,0.597274);
    \drawMark{-1}{0.479367} \drawTriangle{-1}{0.479367}{0}{blue}
    \drawMark{-1}{0.573692} \drawTriangle{-1}{0.573692}{-180}{blue}
    \drawMark{-1}{0.597274} \drawTriangle{-1}{0.597274}{60}{blue}
    \drawMark{-1}{0.620855} \drawTriangle{-1}{0.620855}{60}{blue}
    \drawMark{12}{0.408622} \drawTriangle{12}{0.408622}{0}{blue}
    \drawMark{12}{0.432204} \drawTriangle{12}{0.432204}{0}{blue}
    \drawMark{12}{0.455785} \drawTriangle{12}{0.455785}{0}{blue}
    \drawMark{12}{0.479367} \drawTriangle{12}{0.479367}{0}{blue}
    \drawMark{12}{0.620855} \drawTriangle{12}{0.620855}{0}{blue}
    \drawMark{13}{0.385041} \drawTriangle{13}{0.385041}{0}{blue}
    \drawMark{13}{0.432204} \drawTriangle{13}{0.432204}{0}{blue} \drawTriangle{13}{0.432204}{-60}{blue}
    \drawMark{13}{0.455785} \drawTriangle{13}{0.455785}{0}{blue} \drawTriangle{13}{0.455785}{-60}{blue}
    \drawMark{13}{0.479367} \drawTriangle{13}{0.479367}{0}{blue}
    \drawMark{13}{0.502948} \drawTriangle{13}{0.502948}{0}{blue}
    \drawMark{14}{0.337878} \drawTriangle{14}{0.337878}{0}{blue}
    \drawMark{14}{0.385041} \drawTriangle{14}{0.385041}{0}{blue}
    \drawMark{14}{0.408622} \drawTriangle{14}{0.408622}{0}{blue}
    \drawMark{14}{0.432204} \drawTriangle{14}{0.432204}{0}{blue} \drawTriangle{14}{0.432204}{-60}{blue}
    \drawMark{14}{0.455785} \drawTriangle{14}{0.455785}{-180}{blue} \drawTriangle{14}{0.455785}{0}{blue} \drawTriangle{14}{0.455785}{-60}{blue}
    \drawMark{14}{0.479367} \drawTriangle{14}{0.479367}{0}{blue}
    \drawMark{14}{0.502948} \drawTriangle{14}{0.502948}{0}{blue}
    \drawMark{14}{0.526529} \drawTriangle{14}{0.526529}{0}{blue}
    \drawMark{14}{0.550111} \drawTriangle{14}{0.550111}{0}{blue}
    \drawMark{14}{0.573692} \drawTriangle{14}{0.573692}{-180}{blue} \drawTriangle{14}{0.573692}{0}{blue}
    \drawMark{14}{0.597274} \drawTriangle{14}{0.597274}{-120}{blue} \drawTriangle{14}{0.597274}{-180}{blue} \drawTriangle{14}{0.597274}{0}{blue}
    \drawMark{14}{0.620855} \drawTriangle{14}{0.620855}{-120}{blue}
    \drawMark{14}{0.644437} \drawTriangle{14}{0.644437}{-120}{blue}
    \drawMark{15}{0.361459} \drawTriangle{15}{0.361459}{-180}{blue} \drawTriangle{15}{0.361459}{60}{blue} \drawTriangle{15}{0.361459}{0}{blue}
    \drawMark{15}{0.385041} \drawTriangle{15}{0.385041}{-180}{blue} \drawTriangle{15}{0.385041}{60}{blue} \drawTriangle{15}{0.385041}{0}{blue}
    \drawMark{15}{0.408622} \drawTriangle{15}{0.408622}{-180}{blue} \drawTriangle{15}{0.408622}{60}{blue} \drawTriangle{15}{0.408622}{0}{blue}
    \drawMark{15}{0.432204} \drawTriangle{15}{0.432204}{-180}{blue} \drawTriangle{15}{0.432204}{60}{cyan} \drawTriangle{15}{0.432204}{0}{cyan}
    \drawMark{15}{0.455785} \drawTriangle{15}{0.455785}{-180}{blue} \drawTriangle{15}{0.455785}{60}{cyan} \drawTriangle{15}{0.455785}{0}{blue}
    \drawMark{15}{0.479367} \drawTriangle{15}{0.479367}{-180}{blue} \drawTriangle{15}{0.479367}{60}{blue} \drawTriangle{15}{0.479367}{0}{cyan}
    \drawMark{15}{0.502948} \drawTriangle{15}{0.502948}{-180}{blue} \drawTriangle{15}{0.502948}{60}{blue} \drawTriangle{15}{0.502948}{0}{blue}
    \drawMark{15}{0.526529} \drawTriangle{15}{0.526529}{-180}{blue} \drawTriangle{15}{0.526529}{60}{blue} \drawTriangle{15}{0.526529}{0}{blue}
    \drawMark{15}{0.550111} \drawTriangle{15}{0.550111}{0}{blue}
    \drawMark{15}{0.573692} \drawTriangle{15}{0.573692}{-180}{blue} \drawTriangle{15}{0.573692}{60}{cyan} \drawTriangle{15}{0.573692}{0}{blue}
    \drawMark{15}{0.597274} \drawTriangle{15}{0.597274}{-180}{cyan} \drawTriangle{15}{0.597274}{60}{orange} \drawTriangle{15}{0.597274}{0}{cyan} \drawTriangle{15}{0.597274}{-60}{blue}
    \drawMark{15}{0.620855} \drawTriangle{15}{0.620855}{-180}{blue} \drawTriangle{15}{0.620855}{60}{blue} \drawTriangle{15}{0.620855}{0}{blue} \drawTriangle{15}{0.620855}{-60}{blue}
    \drawMark{15}{0.644437} \drawTriangle{15}{0.644437}{60}{blue}
    \drawMark{15}{0.668018} \drawTriangle{15}{0.668018}{-180}{blue}
    \drawMark{16}{0.361459} \drawTriangle{16}{0.361459}{60}{blue}
    \drawMark{16}{0.385041} \drawTriangle{16}{0.385041}{60}{blue}
    \drawMark{16}{0.408622} \drawTriangle{16}{0.408622}{60}{blue}
    \drawMark{16}{0.432204} \drawTriangle{16}{0.432204}{60}{cyan}
    \drawMark{16}{0.455785} \drawTriangle{16}{0.455785}{60}{cyan}
    \drawMark{16}{0.479367} \drawTriangle{16}{0.479367}{60}{cyan}
    \drawMark{16}{0.502948} \drawTriangle{16}{0.502948}{60}{blue}
    \drawMark{16}{0.526529} \drawTriangle{16}{0.526529}{60}{blue}
    \drawMark{16}{0.550111} \drawTriangle{16}{0.550111}{60}{blue}
    \drawMark{16}{0.573692} \drawTriangle{16}{0.573692}{60}{cyan}
    \drawMark{16}{0.597274} \drawTriangle{16}{0.597274}{60}{red}
    \drawMark{16}{0.620855} \drawTriangle{16}{0.620855}{60}{cyan}
    \end{axis}
  \end{tikzpicture}
\hfill
\message{Instance <mod008>.}
% \message{Instance <mod008> has 391 / 441 cuts.}
% \message{Instance <mod008> has main dimension 319.}
% \message{Instance <mod008> required 1158 nodes.}
% \message{Instance <mod008> has 119 cuts of type 'L'.}
% \message{Instance <mod008> has 272 cuts of type 'R'.}
  
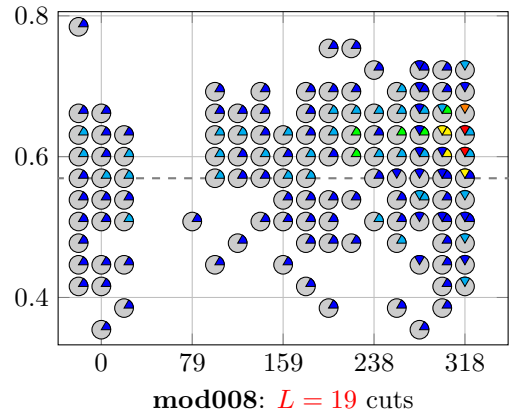
\begin{tikzpicture}
    \begin{axis}[
      enlargelimits = 0.05,
      xlabel = {\textbf{mod008}: \textcolor{red}{$L = 19$} cuts},
      width = 75mm, % old: 66
      height = 60mm, % old: 51
      xtick={0,4,8,12,16},
      xticklabels={0,$79$,$159$,$238$,$318$},
      xmin = -1,
      xmax = 17,
      ymin = 0.3542,
      ymax = 0.784462,
      grid = major,
      colormap = {blueorangered}{rgb255(0cm)=(0,0,255); rgb255(1cm)=(255,128,0); rgb255(2cm)=(255,0,0)},
      colormap name = blueorangered,
    ]
    \draw[black!50!white,thick,dashed] (axis cs:-2,0.569331) -- (axis cs:18,0.569331);
    \drawMark{-1}{0.415666} \drawTriangle{-1}{0.415666}{0}{blue}
    \drawMark{-1}{0.446399} \drawTriangle{-1}{0.446399}{0}{blue}
    \drawMark{-1}{0.477132} \drawTriangle{-1}{0.477132}{0}{blue}
    \drawMark{-1}{0.507865} \drawTriangle{-1}{0.507865}{0}{blue}
    \drawMark{-1}{0.538598} \drawTriangle{-1}{0.538598}{0}{blue}
    \drawMark{-1}{0.569331} \drawTriangle{-1}{0.569331}{0}{blue}
    \drawMark{-1}{0.600064} \drawTriangle{-1}{0.600064}{0}{cyan}
    \drawMark{-1}{0.630797} \drawTriangle{-1}{0.630797}{0}{cyan}
    \drawMark{-1}{0.66153} \drawTriangle{-1}{0.66153}{0}{blue}
    \drawMark{-1}{0.784462} \drawTriangle{-1}{0.784462}{0}{blue}
    \drawMark{0}{0.3542} \drawTriangle{0}{0.3542}{0}{blue}
    \drawMark{0}{0.415666} \drawTriangle{0}{0.415666}{0}{blue}
    \drawMark{0}{0.446399} \drawTriangle{0}{0.446399}{0}{blue}
    \drawMark{0}{0.507865} \drawTriangle{0}{0.507865}{0}{blue}
    \drawMark{0}{0.538598} \drawTriangle{0}{0.538598}{0}{blue}
    \drawMark{0}{0.569331} \drawTriangle{0}{0.569331}{0}{cyan}
    \drawMark{0}{0.600064} \drawTriangle{0}{0.600064}{0}{blue}
    \drawMark{0}{0.630797} \drawTriangle{0}{0.630797}{0}{blue}
    \drawMark{0}{0.66153} \drawTriangle{0}{0.66153}{0}{blue}
    \drawMark{1}{0.384933} \drawTriangle{1}{0.384933}{0}{blue}
    \drawMark{1}{0.446399} \drawTriangle{1}{0.446399}{0}{blue}
    \drawMark{1}{0.507865} \drawTriangle{1}{0.507865}{0}{cyan}
    \drawMark{1}{0.538598} \drawTriangle{1}{0.538598}{0}{blue}
    \drawMark{1}{0.569331} \drawTriangle{1}{0.569331}{0}{cyan}
    \drawMark{1}{0.600064} \drawTriangle{1}{0.600064}{0}{cyan}
    \drawMark{1}{0.630797} \drawTriangle{1}{0.630797}{0}{blue}
    \drawMark{4}{0.507865} \drawTriangle{4}{0.507865}{0}{blue}
    \drawMark{5}{0.446399} \drawTriangle{5}{0.446399}{0}{blue}
    \drawMark{5}{0.569331} \drawTriangle{5}{0.569331}{0}{cyan}
    \drawMark{5}{0.600064} \drawTriangle{5}{0.600064}{0}{cyan}
    \drawMark{5}{0.630797} \drawTriangle{5}{0.630797}{0}{cyan}
    \drawMark{5}{0.66153} \drawTriangle{5}{0.66153}{0}{blue}
    \drawMark{5}{0.692263} \drawTriangle{5}{0.692263}{0}{blue}
    \drawMark{6}{0.477132} \drawTriangle{6}{0.477132}{0}{blue}
    \drawMark{6}{0.569331} \drawTriangle{6}{0.569331}{0}{blue}
    \drawMark{6}{0.600064} \drawTriangle{6}{0.600064}{0}{blue}
    \drawMark{6}{0.630797} \drawTriangle{6}{0.630797}{0}{cyan}
    \drawMark{6}{0.66153} \drawTriangle{6}{0.66153}{0}{blue}
    \drawMark{7}{0.507865} \drawTriangle{7}{0.507865}{0}{blue}
    \drawMark{7}{0.569331} \drawTriangle{7}{0.569331}{0}{blue}
    \drawMark{7}{0.600064} \drawTriangle{7}{0.600064}{0}{cyan}
    \drawMark{7}{0.630797} \drawTriangle{7}{0.630797}{0}{blue}
    \drawMark{7}{0.66153} \drawTriangle{7}{0.66153}{0}{blue}
    \drawMark{7}{0.692263} \drawTriangle{7}{0.692263}{0}{blue}
    \drawMark{8}{0.446399} \drawTriangle{8}{0.446399}{0}{blue}
    \drawMark{8}{0.507865} \drawTriangle{8}{0.507865}{0}{blue}
    \drawMark{8}{0.538598} \drawTriangle{8}{0.538598}{0}{blue}
    \drawMark{8}{0.569331} \drawTriangle{8}{0.569331}{0}{cyan}
    \drawMark{8}{0.600064} \drawTriangle{8}{0.600064}{0}{blue}
    \drawMark{8}{0.630797} \drawTriangle{8}{0.630797}{0}{cyan}
    \drawMark{9}{0.415666} \drawTriangle{9}{0.415666}{0}{blue}
    \drawMark{9}{0.477132} \drawTriangle{9}{0.477132}{0}{blue}
    \drawMark{9}{0.507865} \drawTriangle{9}{0.507865}{0}{blue}
    \drawMark{9}{0.538598} \drawTriangle{9}{0.538598}{0}{blue}
    \drawMark{9}{0.569331} \drawTriangle{9}{0.569331}{0}{cyan}
    \drawMark{9}{0.600064} \drawTriangle{9}{0.600064}{0}{cyan}
    \drawMark{9}{0.630797} \drawTriangle{9}{0.630797}{0}{blue}
    \drawMark{9}{0.66153} \drawTriangle{9}{0.66153}{0}{blue}
    \drawMark{9}{0.692263} \drawTriangle{9}{0.692263}{0}{blue}
    \drawMark{10}{0.384933} \drawTriangle{10}{0.384933}{0}{blue}
    \drawMark{10}{0.477132} \drawTriangle{10}{0.477132}{0}{blue}
    \drawMark{10}{0.507865} \drawTriangle{10}{0.507865}{0}{blue}
    \drawMark{10}{0.538598} \drawTriangle{10}{0.538598}{0}{blue}
    \drawMark{10}{0.600064} \drawTriangle{10}{0.600064}{0}{blue}
    \drawMark{10}{0.630797} \drawTriangle{10}{0.630797}{0}{cyan}
    \drawMark{10}{0.66153} \drawTriangle{10}{0.66153}{0}{cyan}
    \drawMark{10}{0.692263} \drawTriangle{10}{0.692263}{0}{blue}
    \drawMark{10}{0.753729} \drawTriangle{10}{0.753729}{0}{blue}
    \drawMark{11}{0.477132} \drawTriangle{11}{0.477132}{0}{blue}
    \drawMark{11}{0.538598} \drawTriangle{11}{0.538598}{0}{blue}
    \drawMark{11}{0.600064} \drawTriangle{11}{0.600064}{0}{green}
    \drawMark{11}{0.630797} \drawTriangle{11}{0.630797}{0}{green}
    \drawMark{11}{0.66153} \drawTriangle{11}{0.66153}{0}{cyan}
    \drawMark{11}{0.692263} \drawTriangle{11}{0.692263}{0}{blue}
    \drawMark{11}{0.753729} \drawTriangle{11}{0.753729}{0}{blue}
    \drawMark{12}{0.507865} \drawTriangle{12}{0.507865}{0}{cyan}
    \drawMark{12}{0.569331} \drawTriangle{12}{0.569331}{0}{blue}
    \drawMark{12}{0.600064} \drawTriangle{12}{0.600064}{0}{cyan}
    \drawMark{12}{0.630797} \drawTriangle{12}{0.630797}{0}{blue}
    \drawMark{12}{0.66153} \drawTriangle{12}{0.66153}{0}{cyan}
    \drawMark{12}{0.722996} \drawTriangle{12}{0.722996}{0}{blue}
    \drawMark{13}{0.384933} \drawTriangle{13}{0.384933}{0}{blue}
    \drawMark{13}{0.477132} \drawTriangle{13}{0.477132}{0}{cyan}
    \drawMark{13}{0.507865} \drawTriangle{13}{0.507865}{0}{blue}
    \drawMark{13}{0.538598} \drawTriangle{13}{0.538598}{0}{blue}
    \drawMark{13}{0.569331} \drawTriangle{13}{0.569331}{60}{blue}
    \drawMark{13}{0.600064} \drawTriangle{13}{0.600064}{0}{cyan}
    \drawMark{13}{0.630797} \drawTriangle{13}{0.630797}{0}{green}
    \drawMark{13}{0.66153} \drawTriangle{13}{0.66153}{0}{cyan}
    \drawMark{13}{0.692263} \drawTriangle{13}{0.692263}{0}{cyan}
    \drawMark{14}{0.3542} \drawTriangle{14}{0.3542}{0}{blue}
    \drawMark{14}{0.446399} \drawTriangle{14}{0.446399}{60}{blue}
    \drawMark{14}{0.507865} \drawTriangle{14}{0.507865}{60}{blue}
    \drawMark{14}{0.538598} \drawTriangle{14}{0.538598}{60}{cyan} \drawTriangle{14}{0.538598}{0}{cyan}
    \drawMark{14}{0.569331} \drawTriangle{14}{0.569331}{60}{blue}
    \drawMark{14}{0.600064} \drawTriangle{14}{0.600064}{60}{blue} \drawTriangle{14}{0.600064}{0}{cyan}
    \drawMark{14}{0.630797} \drawTriangle{14}{0.630797}{60}{blue} \drawTriangle{14}{0.630797}{0}{green}
    \drawMark{14}{0.66153} \drawTriangle{14}{0.66153}{60}{blue} \drawTriangle{14}{0.66153}{0}{cyan}
    \drawMark{14}{0.692263} \drawTriangle{14}{0.692263}{60}{blue} \drawTriangle{14}{0.692263}{0}{blue}
    \drawMark{14}{0.722996} \drawTriangle{14}{0.722996}{60}{blue} \drawTriangle{14}{0.722996}{0}{blue}
    \drawMark{15}{0.384933} \drawTriangle{15}{0.384933}{0}{blue}
    \drawMark{15}{0.415666} \drawTriangle{15}{0.415666}{0}{blue}
    \drawMark{15}{0.446399} \drawTriangle{15}{0.446399}{0}{blue}
    \drawMark{15}{0.477132} \drawTriangle{15}{0.477132}{0}{blue}
    \drawMark{15}{0.507865} \drawTriangle{15}{0.507865}{60}{blue} \drawTriangle{15}{0.507865}{0}{blue}
    \drawMark{15}{0.538598} \drawTriangle{15}{0.538598}{0}{blue}
    \drawMark{15}{0.569331} \drawTriangle{15}{0.569331}{60}{blue} \drawTriangle{15}{0.569331}{0}{blue}
    \drawMark{15}{0.600064} \drawTriangle{15}{0.600064}{60}{blue} \drawTriangle{15}{0.600064}{0}{yellow}
    \drawMark{15}{0.630797} \drawTriangle{15}{0.630797}{60}{yellow} \drawTriangle{15}{0.630797}{0}{yellow}
    \drawMark{15}{0.66153} \drawTriangle{15}{0.66153}{60}{cyan} \drawTriangle{15}{0.66153}{0}{green}
    \drawMark{15}{0.692263} \drawTriangle{15}{0.692263}{0}{blue}
    \drawMark{15}{0.722996} \drawTriangle{15}{0.722996}{0}{blue}
    \drawMark{16}{0.415666} \drawTriangle{16}{0.415666}{60}{cyan}
    \drawMark{16}{0.446399} \drawTriangle{16}{0.446399}{60}{blue}
    \drawMark{16}{0.477132} \drawTriangle{16}{0.477132}{60}{cyan}
    \drawMark{16}{0.507865} \drawTriangle{16}{0.507865}{60}{blue} \drawTriangle{16}{0.507865}{0}{blue}
    \drawMark{16}{0.538598} \drawTriangle{16}{0.538598}{60}{cyan}
    \drawMark{16}{0.569331} \drawTriangle{16}{0.569331}{60}{yellow} \drawTriangle{16}{0.569331}{0}{blue}
    \drawMark{16}{0.600064} \drawTriangle{16}{0.600064}{60}{red} \drawTriangle{16}{0.600064}{0}{cyan}
    \drawMark{16}{0.630797} \drawTriangle{16}{0.630797}{60}{red} \drawTriangle{16}{0.630797}{0}{cyan}
    \drawMark{16}{0.66153} \drawTriangle{16}{0.66153}{60}{orange}
    \drawMark{16}{0.692263} \drawTriangle{16}{0.692263}{60}{blue}
    \drawMark{16}{0.722996} \drawTriangle{16}{0.722996}{60}{cyan}
    \end{axis}
  \end{tikzpicture}

\bigskip

The dashed line for \texttt{harp2} shows that adding a single cut does not necessarily help in branch-and-bound, which may be due to side-effects such as different branching decisions.
For some instances, such as \texttt{mod008}, the cuts' face dimensions are well distributed.
In contrast to this, some instances exhibit only very few distinct dimension values, e.g., only non-supporting cuts for \texttt{qiu}.
Interestingly, the 6246 cuts for \texttt{seymour} induce only empty faces as well as faces with dimensions between $1218$ and $1254$.
This can partially be explained via the cut classes.
On the one hand, all generated strengthened Chv{\'a}tal-Gomory cuts are non-supporting.
On the other hand, some of the c-MIR cuts and $\{0,1/2\}$-cuts are non-supporting while others induce faces of very high dimension.

\bigskip

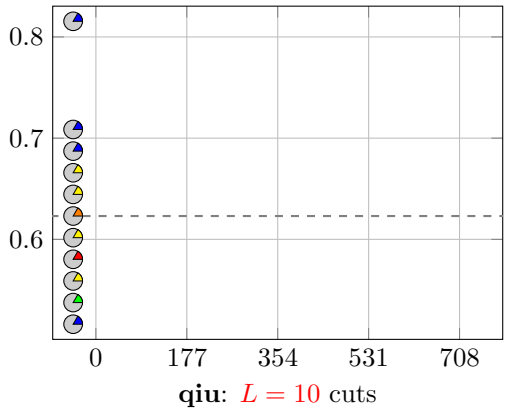
\message{Instance <qiu>.}
% \message{Instance <qiu> has 51 / 63 cuts.}
% \message{Instance <qiu> has main dimension 709.}
% \message{Instance <qiu> required 10406 nodes.}
% \message{Instance <qiu> has 51 cuts of type 'R'.}
  \begin{tikzpicture}
    \begin{axis}[
      enlargelimits = 0.05,
      xlabel = {\textbf{qiu}: \textcolor{red}{$L = 10$} cuts},
      width = 75mm, % old: 66
      height = 60mm, % old: 51
      xtick={0,4,8,12,16},
      xticklabels={0,$177$,$354$,$531$,$708$},
      xmin = -1,
      xmax = 17,
      ymin = 0.516131,
      ymax = 0.815454,
      grid = major,
      colormap = {blueorangered}{rgb255(0cm)=(0,0,255); rgb255(1cm)=(255,128,0); rgb255(2cm)=(255,0,0)},
      colormap name = blueorangered,
    ]
    \draw[black!50!white,thick,dashed] (axis cs:-2,0.623032) -- (axis cs:18,0.623032);
    \drawMark{-1}{0.516131} \drawTriangle{-1}{0.516131}{0}{blue}
    \drawMark{-1}{0.537511} \drawTriangle{-1}{0.537511}{0}{green}
    \drawMark{-1}{0.558891} \drawTriangle{-1}{0.558891}{0}{yellow}
    \drawMark{-1}{0.580272} \drawTriangle{-1}{0.580272}{0}{red}
    \drawMark{-1}{0.601652} \drawTriangle{-1}{0.601652}{0}{yellow}
    \drawMark{-1}{0.623032} \drawTriangle{-1}{0.623032}{0}{orange}
    \drawMark{-1}{0.644412} \drawTriangle{-1}{0.644412}{0}{yellow}
    \drawMark{-1}{0.665793} \drawTriangle{-1}{0.665793}{0}{yellow}
    \drawMark{-1}{0.687173} \drawTriangle{-1}{0.687173}{0}{blue}
    \drawMark{-1}{0.708553} \drawTriangle{-1}{0.708553}{0}{blue}
    \drawMark{-1}{0.815454} \drawTriangle{-1}{0.815454}{0}{blue}
    \end{axis}
  \end{tikzpicture}
\hfill
\message{Instance <seymour>.}
% \message{Instance <seymour> has 6237 / 6246 cuts.}
% \message{Instance <seymour> has main dimension 1255.}
% \message{Instance <seymour> required 9 nodes.}
% \message{Instance <seymour> has 656 cuts of type 'R'.}
% \message{Instance <seymour> has 53 cuts of type 'Z'.}
% \message{Instance <seymour> has 5528 cuts of type 'C'.}
  
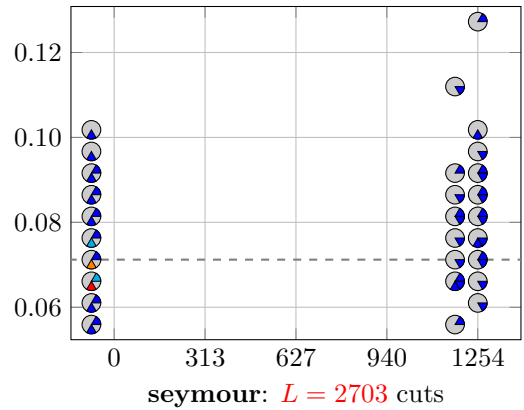
\begin{tikzpicture}
    \begin{axis}[
      enlargelimits = 0.05,
      xlabel = {\textbf{seymour}: \textcolor{red}{$L = 2703$} cuts},
      width = 75mm, % old: 66
      height = 60mm, % old: 51
      xtick={0,4,8,12,16},
      xticklabels={0,$313$,$627$,$940$,$1254$},
      ytick={0.06, 0.08, 0.10, 0.12},
      yticklabels={0.06, 0.08, 0.10, 0.12},
      xmin = -1,
      xmax = 17,
      ymin = 0.055919,
      ymax = 0.127267,
      grid = major,
      colormap = {blueorangered}{rgb255(0cm)=(0,0,255); rgb255(1cm)=(255,128,0); rgb255(2cm)=(255,0,0)},
      colormap name = blueorangered,
    ]
    \draw[black!50!white,thick,dashed] (axis cs:-2,0.0712079) -- (axis cs:18,0.0712079);
    \drawMark{-1}{0.055919} \drawTriangle{-1}{0.055919}{-120}{blue} \drawTriangle{-1}{0.055919}{0}{blue}
    \drawMark{-1}{0.0610153} \drawTriangle{-1}{0.0610153}{-120}{blue} \drawTriangle{-1}{0.0610153}{0}{blue}
    \drawMark{-1}{0.0661116} \drawTriangle{-1}{0.0661116}{-120}{red} \drawTriangle{-1}{0.0661116}{0}{cyan}
    \drawMark{-1}{0.0712079} \drawTriangle{-1}{0.0712079}{-120}{orange} \drawTriangle{-1}{0.0712079}{0}{blue}
    \drawMark{-1}{0.0763041} \drawTriangle{-1}{0.0763041}{-120}{cyan} \drawTriangle{-1}{0.0763041}{0}{blue}
    \drawMark{-1}{0.0814004} \drawTriangle{-1}{0.0814004}{-120}{blue} \drawTriangle{-1}{0.0814004}{0}{blue}
    \drawMark{-1}{0.0864967} \drawTriangle{-1}{0.0864967}{-120}{blue} \drawTriangle{-1}{0.0864967}{0}{blue}
    \drawMark{-1}{0.091593} \drawTriangle{-1}{0.091593}{-120}{blue} \drawTriangle{-1}{0.091593}{0}{blue}
    \drawMark{-1}{0.0966893} \drawTriangle{-1}{0.0966893}{-120}{blue}
    \drawMark{-1}{0.101786} \drawTriangle{-1}{0.101786}{-120}{blue}
    \drawMark{15}{0.055919} \drawTriangle{15}{0.055919}{0}{blue}
    \drawMark{15}{0.0661116} \drawTriangle{15}{0.0661116}{-120}{blue} \drawTriangle{15}{0.0661116}{0}{blue} \drawTriangle{15}{0.0661116}{-60}{blue}
    \drawMark{15}{0.0712079} \drawTriangle{15}{0.0712079}{-60}{blue}
    \drawMark{15}{0.0763041} \drawTriangle{15}{0.0763041}{-60}{blue}
    \drawMark{15}{0.0814004} \drawTriangle{15}{0.0814004}{0}{blue} \drawTriangle{15}{0.0814004}{-60}{blue}
    \drawMark{15}{0.0864967} \drawTriangle{15}{0.0864967}{-60}{blue}
    \drawMark{15}{0.091593} \drawTriangle{15}{0.091593}{0}{blue}
    \drawMark{15}{0.111978} \drawTriangle{15}{0.111978}{-60}{blue}
    \drawMark{16}{0.0610153} \drawTriangle{16}{0.0610153}{-60}{blue}
    \drawMark{16}{0.0661116} \drawTriangle{16}{0.0661116}{-60}{blue}
    \drawMark{16}{0.0712079} \drawTriangle{16}{0.0712079}{0}{blue} \drawTriangle{16}{0.0712079}{-60}{blue}
    \drawMark{16}{0.0763041} \drawTriangle{16}{0.0763041}{-120}{blue} \drawTriangle{16}{0.0763041}{-60}{blue}
    \drawMark{16}{0.0814004} \drawTriangle{16}{0.0814004}{0}{blue} \drawTriangle{16}{0.0814004}{-60}{blue}
    \drawMark{16}{0.0864967} \drawTriangle{16}{0.0864967}{0}{blue} \drawTriangle{16}{0.0864967}{-60}{blue}
    \drawMark{16}{0.091593} \drawTriangle{16}{0.091593}{0}{blue} \drawTriangle{16}{0.091593}{-60}{blue}
    \drawMark{16}{0.0966893} \drawTriangle{16}{0.0966893}{-60}{blue}
    \drawMark{16}{0.101786} \drawTriangle{16}{0.101786}{-120}{blue}
    \drawMark{16}{0.127267} \drawTriangle{16}{0.127267}{0}{blue}
    \end{axis}
  \end{tikzpicture}

\bigskip

In general we don't see an indication that cuts from certain classes induce higher dimensional faces than others.
At first glance, such a pattern is apparent for \texttt{misc07} (at dimensions 130--160), however one has to keep in mind that these blue segments constitute a minority of the cuts.
In line with that, the majority of the cuts for \texttt{fiber} is concentrated around dimension 900 with a closed gap similar to that without cuts.

\bigskip

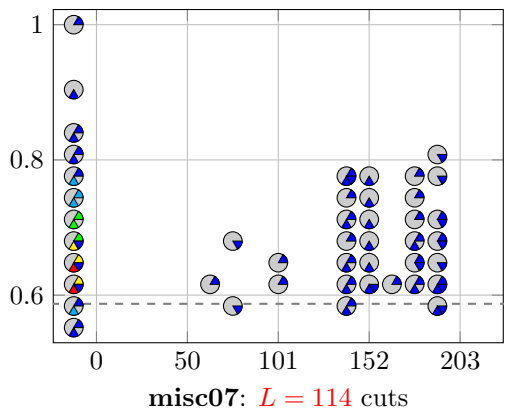
\message{Instance <misc07>.}
% \message{Instance <misc07> has 759 / 808 cuts.}
% \message{Instance <misc07> has main dimension 204.}
% \message{Instance <misc07> required 22839 nodes.}
% \message{Instance <misc07> has 286 cuts of type 'R'.}
% \message{Instance <misc07> has 34 cuts of type 'Z'.}
% \message{Instance <misc07> has 440 cuts of type 'C'.}
\hspace{-2mm}
  \begin{tikzpicture}
    \begin{axis}[
      enlargelimits = 0.05,
      xlabel = {\textbf{misc07}: \textcolor{red}{$L = 114$} cuts},
      width = 75mm, % old: 66
      height = 60mm, % old: 51
      xtick={0,4,8,12,16},
      xticklabels={0,$50$,$101$,$152$,$203$},
      xmin = -1,
      xmax = 17,
      ymin = 0.551971,
      ymax = 1,
      grid = major,
      colormap = {blueorangered}{rgb255(0cm)=(0,0,255); rgb255(1cm)=(255,128,0); rgb255(2cm)=(255,0,0)},
      colormap name = blueorangered,
    ]
    \draw[black!50!white,thick,dashed] (axis cs:-2,0.587173) -- (axis cs:18,0.587173);
    \drawMark{-1}{0.551971} \drawTriangle{-1}{0.551971}{-120}{blue} \drawTriangle{-1}{0.551971}{0}{blue}
    \drawMark{-1}{0.583973} \drawTriangle{-1}{0.583973}{-120}{cyan} \drawTriangle{-1}{0.583973}{0}{blue}
    \drawMark{-1}{0.615975} \drawTriangle{-1}{0.615975}{-120}{red} \drawTriangle{-1}{0.615975}{0}{yellow} \drawTriangle{-1}{0.615975}{-60}{blue}
    \drawMark{-1}{0.647977} \drawTriangle{-1}{0.647977}{-120}{red} \drawTriangle{-1}{0.647977}{0}{yellow} \drawTriangle{-1}{0.647977}{-60}{blue}
    \drawMark{-1}{0.679979} \drawTriangle{-1}{0.679979}{-120}{yellow} \drawTriangle{-1}{0.679979}{0}{green} \drawTriangle{-1}{0.679979}{-60}{blue}
    \drawMark{-1}{0.711981} \drawTriangle{-1}{0.711981}{-120}{green} \drawTriangle{-1}{0.711981}{0}{green}
    \drawMark{-1}{0.743983} \drawTriangle{-1}{0.743983}{-120}{cyan} \drawTriangle{-1}{0.743983}{0}{cyan}
    \drawMark{-1}{0.775985} \drawTriangle{-1}{0.775985}{-120}{cyan} \drawTriangle{-1}{0.775985}{0}{blue}
    \drawMark{-1}{0.807988} \drawTriangle{-1}{0.807988}{-120}{blue} \drawTriangle{-1}{0.807988}{0}{blue}
    \drawMark{-1}{0.83999} \drawTriangle{-1}{0.83999}{-120}{blue} \drawTriangle{-1}{0.83999}{0}{blue}
    \drawMark{-1}{0.903994} \drawTriangle{-1}{0.903994}{-120}{blue}
    \drawMark{-1}{1} \drawTriangle{-1}{1}{0}{blue}
    \drawMark{5}{0.615975} \drawTriangle{5}{0.615975}{0}{blue}
    \drawMark{6}{0.583973} \drawTriangle{6}{0.583973}{-60}{blue}
    \drawMark{6}{0.679979} \drawTriangle{6}{0.679979}{-60}{blue}
    \drawMark{8}{0.615975} \drawTriangle{8}{0.615975}{0}{blue}
    \drawMark{8}{0.647977} \drawTriangle{8}{0.647977}{0}{blue}
    \drawMark{11}{0.583973} \drawTriangle{11}{0.583973}{-120}{blue} \drawTriangle{11}{0.583973}{0}{blue}
    \drawMark{11}{0.615975} \drawTriangle{11}{0.615975}{-120}{blue} \drawTriangle{11}{0.615975}{0}{blue}
    \drawMark{11}{0.647977} \drawTriangle{11}{0.647977}{-120}{blue} \drawTriangle{11}{0.647977}{0}{blue}
    \drawMark{11}{0.679979} \drawTriangle{11}{0.679979}{0}{blue}
    \drawMark{11}{0.711981} \drawTriangle{11}{0.711981}{-120}{blue} \drawTriangle{11}{0.711981}{0}{blue}
    \drawMark{11}{0.743983} \drawTriangle{11}{0.743983}{0}{blue}
    \drawMark{11}{0.775985} \drawTriangle{11}{0.775985}{-120}{blue} \drawTriangle{11}{0.775985}{0}{blue} \drawTriangle{11}{0.775985}{-60}{blue}
    \drawMark{12}{0.615975} \drawTriangle{12}{0.615975}{-120}{blue} \drawTriangle{12}{0.615975}{-60}{blue}
    \drawMark{12}{0.647977} \drawTriangle{12}{0.647977}{-120}{blue}
    \drawMark{12}{0.679979} \drawTriangle{12}{0.679979}{-120}{blue}
    \drawMark{12}{0.711981} \drawTriangle{12}{0.711981}{-120}{blue}
    \drawMark{12}{0.743983} \drawTriangle{12}{0.743983}{-120}{blue}
    \drawMark{12}{0.775985} \drawTriangle{12}{0.775985}{-120}{blue}
    \drawMark{13}{0.615975} \drawTriangle{13}{0.615975}{0}{blue}
    \drawMark{14}{0.615975} \drawTriangle{14}{0.615975}{-120}{blue} \drawTriangle{14}{0.615975}{0}{blue}
    \drawMark{14}{0.647977} \drawTriangle{14}{0.647977}{0}{blue} \drawTriangle{14}{0.647977}{-60}{blue}
    \drawMark{14}{0.679979} \drawTriangle{14}{0.679979}{-120}{blue} \drawTriangle{14}{0.679979}{0}{blue}
    \drawMark{14}{0.711981} \drawTriangle{14}{0.711981}{0}{blue}
    \drawMark{14}{0.743983} \drawTriangle{14}{0.743983}{0}{blue}
    \drawMark{14}{0.775985} \drawTriangle{14}{0.775985}{0}{blue}
    \drawMark{15}{0.583973} \drawTriangle{15}{0.583973}{-120}{blue} \drawTriangle{15}{0.583973}{-60}{blue}
    \drawMark{15}{0.615975} \drawTriangle{15}{0.615975}{-120}{blue} \drawTriangle{15}{0.615975}{0}{blue} \drawTriangle{15}{0.615975}{-60}{blue}
    \drawMark{15}{0.647977} \drawTriangle{15}{0.647977}{-60}{blue}
    \drawMark{15}{0.679979} \drawTriangle{15}{0.679979}{0}{blue} \drawTriangle{15}{0.679979}{-60}{blue}
    \drawMark{15}{0.711981} \drawTriangle{15}{0.711981}{0}{blue} \drawTriangle{15}{0.711981}{-60}{blue}
    \drawMark{15}{0.775985} \drawTriangle{15}{0.775985}{-60}{blue}
    \drawMark{15}{0.807988} \drawTriangle{15}{0.807988}{-60}{blue}
    \end{axis}
  \end{tikzpicture}
\hfill
\message{Instance <fiber>.}
% \message{Instance <fiber> has 493 / 533 cuts.}
% \message{Instance <fiber> has main dimension 946.}
% \message{Instance <fiber> required 23395 nodes.}
% \message{Instance <fiber> has 127 cuts of type 'L'.}
% \message{Instance <fiber> has 139 cuts of type 'R'.}
% \message{Instance <fiber> has 7 cuts of type 'Z'.}
% \message{Instance <fiber> has 4 cuts of type 'C'.}
% \message{Instance <fiber> has 9 cuts of type 'F'.}
% \message{Instance <fiber> has 207 cuts of type 'M'.}
  
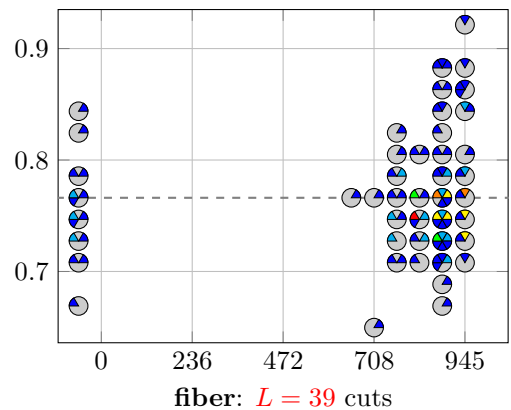
\begin{tikzpicture}
    \begin{axis}[
      enlargelimits = 0.05,
      xlabel = {\textbf{fiber}: \textcolor{red}{$L = 39$} cuts},
      width = 75mm, % old: 66
      height = 60mm, % old: 51
      xtick={0,4,8,12,16},
      xticklabels={0,$236$,$472$,$708$,$945$},
      xmin = -1,
      xmax = 17,
      ymin = 0.649644,
      ymax = 0.921407,
      grid = major,
      colormap = {blueorangered}{rgb255(0cm)=(0,0,255); rgb255(1cm)=(255,128,0); rgb255(2cm)=(255,0,0)},
      colormap name = blueorangered,
    ]
    \draw[black!50!white,thick,dashed] (axis cs:-2,0.766114) -- (axis cs:18,0.766114);
    \drawMark{-1}{0.669056} \drawTriangle{-1}{0.669056}{-240}{blue}
    \drawMark{-1}{0.707879} \drawTriangle{-1}{0.707879}{-240}{blue} \drawTriangle{-1}{0.707879}{0}{blue}
    \drawMark{-1}{0.727291} \drawTriangle{-1}{0.727291}{-240}{cyan} \drawTriangle{-1}{0.727291}{0}{blue}
    \drawMark{-1}{0.746702} \drawTriangle{-1}{0.746702}{-180}{blue} \drawTriangle{-1}{0.746702}{-240}{cyan} \drawTriangle{-1}{0.746702}{0}{blue}
    \drawMark{-1}{0.766114} \drawTriangle{-1}{0.766114}{-180}{blue} \drawTriangle{-1}{0.766114}{-240}{cyan} \drawTriangle{-1}{0.766114}{0}{blue}
    \drawMark{-1}{0.785525} \drawTriangle{-1}{0.785525}{-240}{blue} \drawTriangle{-1}{0.785525}{0}{blue}
    \drawMark{-1}{0.824349} \drawTriangle{-1}{0.824349}{0}{blue}
    \drawMark{-1}{0.84376} \drawTriangle{-1}{0.84376}{0}{blue}
    \drawMark{11}{0.766114} \drawTriangle{11}{0.766114}{0}{blue}
    \drawMark{12}{0.649644} \drawTriangle{12}{0.649644}{0}{blue}
    \drawMark{12}{0.766114} \drawTriangle{12}{0.766114}{0}{blue}
    \drawMark{13}{0.707879} \drawTriangle{13}{0.707879}{-240}{blue} \drawTriangle{13}{0.707879}{0}{blue}
    \drawMark{13}{0.727291} \drawTriangle{13}{0.727291}{-240}{cyan}
    \drawMark{13}{0.746702} \drawTriangle{13}{0.746702}{-240}{cyan} \drawTriangle{13}{0.746702}{0}{blue}
    \drawMark{13}{0.766114} \drawTriangle{13}{0.766114}{-240}{blue} \drawTriangle{13}{0.766114}{0}{blue}
    \drawMark{13}{0.785525} \drawTriangle{13}{0.785525}{-240}{blue} \drawTriangle{13}{0.785525}{0}{cyan}
    \drawMark{13}{0.804937} \drawTriangle{13}{0.804937}{0}{blue}
    \drawMark{13}{0.824349} \drawTriangle{13}{0.824349}{0}{blue}
    \drawMark{14}{0.707879} \drawTriangle{14}{0.707879}{-240}{blue}
    \drawMark{14}{0.727291} \drawTriangle{14}{0.727291}{-240}{blue} \drawTriangle{14}{0.727291}{0}{cyan}
    \drawMark{14}{0.746702} \drawTriangle{14}{0.746702}{-180}{blue} \drawTriangle{14}{0.746702}{-240}{red} \drawTriangle{14}{0.746702}{0}{cyan}
    \drawMark{14}{0.766114} \drawTriangle{14}{0.766114}{-240}{green} \drawTriangle{14}{0.766114}{0}{blue}
    \drawMark{14}{0.804937} \drawTriangle{14}{0.804937}{-240}{blue} \drawTriangle{14}{0.804937}{0}{blue}
    \drawMark{15}{0.669056} \drawTriangle{15}{0.669056}{0}{blue}
    \drawMark{15}{0.688467} \drawTriangle{15}{0.688467}{0}{blue}
    \drawMark{15}{0.707879} \drawTriangle{15}{0.707879}{-180}{blue} \drawTriangle{15}{0.707879}{60}{blue} \drawTriangle{15}{0.707879}{-240}{blue} \drawTriangle{15}{0.707879}{0}{cyan}
    \drawMark{15}{0.727291} \drawTriangle{15}{0.727291}{-120}{blue} \drawTriangle{15}{0.727291}{-180}{blue} \drawTriangle{15}{0.727291}{60}{cyan} \drawTriangle{15}{0.727291}{-240}{green} \drawTriangle{15}{0.727291}{0}{cyan} \drawTriangle{15}{0.727291}{-60}{blue}
    \drawMark{15}{0.746702} \drawTriangle{15}{0.746702}{-120}{blue} \drawTriangle{15}{0.746702}{-180}{blue} \drawTriangle{15}{0.746702}{60}{cyan} \drawTriangle{15}{0.746702}{-240}{yellow} \drawTriangle{15}{0.746702}{0}{yellow} \drawTriangle{15}{0.746702}{-60}{blue}
    \drawMark{15}{0.766114} \drawTriangle{15}{0.766114}{-120}{blue} \drawTriangle{15}{0.766114}{60}{cyan} \drawTriangle{15}{0.766114}{-240}{orange} \drawTriangle{15}{0.766114}{0}{yellow} \drawTriangle{15}{0.766114}{-60}{blue}
    \drawMark{15}{0.785525} \drawTriangle{15}{0.785525}{60}{blue} \drawTriangle{15}{0.785525}{-240}{blue} \drawTriangle{15}{0.785525}{0}{cyan}
    \drawMark{15}{0.804937} \drawTriangle{15}{0.804937}{-240}{blue} \drawTriangle{15}{0.804937}{0}{blue}
    \drawMark{15}{0.824349} \drawTriangle{15}{0.824349}{-240}{blue}
    \drawMark{15}{0.84376} \drawTriangle{15}{0.84376}{60}{blue} \drawTriangle{15}{0.84376}{-240}{blue}
    \drawMark{15}{0.863172} \drawTriangle{15}{0.863172}{-240}{blue} \drawTriangle{15}{0.863172}{0}{blue}
    \drawMark{15}{0.882584} \drawTriangle{15}{0.882584}{60}{blue} \drawTriangle{15}{0.882584}{-240}{blue} \drawTriangle{15}{0.882584}{0}{blue}
    \drawMark{16}{0.707879} \drawTriangle{16}{0.707879}{60}{blue}
    \drawMark{16}{0.727291} \drawTriangle{16}{0.727291}{60}{yellow} \drawTriangle{16}{0.727291}{-240}{blue}
    \drawMark{16}{0.746702} \drawTriangle{16}{0.746702}{60}{yellow} \drawTriangle{16}{0.746702}{-240}{blue}
    \drawMark{16}{0.766114} \drawTriangle{16}{0.766114}{60}{orange} \drawTriangle{16}{0.766114}{-240}{blue}
    \drawMark{16}{0.785525} \drawTriangle{16}{0.785525}{60}{cyan} \drawTriangle{16}{0.785525}{-240}{blue}
    \drawMark{16}{0.804937} \drawTriangle{16}{0.804937}{0}{blue}
    \drawMark{16}{0.84376} \drawTriangle{16}{0.84376}{60}{cyan} \drawTriangle{16}{0.84376}{0}{blue}
    \drawMark{16}{0.863172} \drawTriangle{16}{0.863172}{-180}{blue} \drawTriangle{16}{0.863172}{60}{blue} \drawTriangle{16}{0.863172}{-240}{blue}
    \drawMark{16}{0.882584} \drawTriangle{16}{0.882584}{60}{blue}
    \drawMark{16}{0.921407} \drawTriangle{16}{0.921407}{60}{blue}
    \end{axis}
  \end{tikzpicture}

\bigskip

Despite the heterogeneity of the results, one observation is common to many instances: the distribution of the face dimensions is biased towards $-1$ and high-dimensions, i.e., not many cuts inducing low dimensional faces are generated.

A corresponding histogram is depicted in Fig.~\ref{fig_dimension_distribution}.
We conjecture that the high dimensions occur because lifting and strengthening techniques for cutting planes are quite evolved.

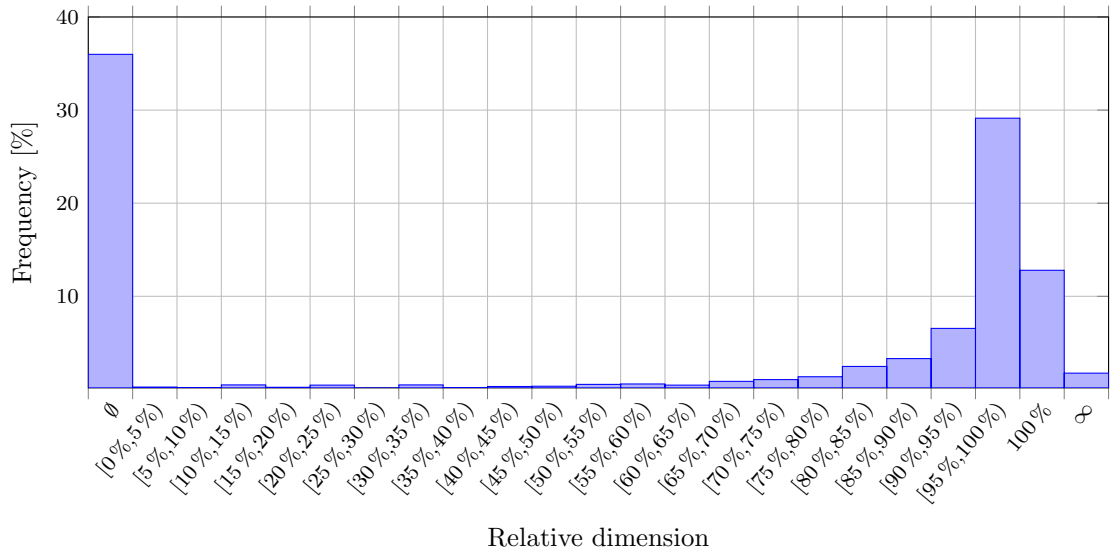
\begin{figure}[htb]
  \begin{tikzpicture}
    \begin{axis}[
      enlargelimits = false,
      width = 150mm,
      height = 65mm,
      ylabel = {Frequency [\%]},
      xlabel = {Relative dimension},
      xlabel style = {yshift=0pt},
      ylabel style = {yshift=0pt},
      ymax = 40,
      ytick = {0,10,20,30,40},
      xticklabels = {$\emptyset$,{[0\,\%,5\,\%)},{[5\,\%,10\,\%)},{[10\,\%,15\,\%)},{[15\,\%,20\,\%)},{[20\,\%,25\,\%)},{[25\,\%,30\,\%)},{[30\,\%,35\,\%)},{[35\,\%,40\,\%)},{[40\,\%,45\,\%)},{[45\,\%,50\,\%)},{[50\,\%,55\,\%)},{[55\,\%,60\,\%)},{[60\,\%,65\,\%)},{[65\,\%,70\,\%)},{[70\,\%,75\,\%)},{[75\,\%,80\,\%)},{[80\,\%,85\,\%)},{[85\,\%,90\,\%)},{[90\,\%,95\,\%)},{[95\,\%,100\,\%)},100\,\%,$\infty$},
      y tick label style={font=\footnotesize},
      x tick label style={font=\footnotesize, rotate=50, anchor=east, xshift=1mm, yshift=-1mm},
      grid = major,
      ybar interval,
    ]
    \addplot coordinates {
 (-1,35.9886) (0,0.253515) (1,0.172633) (2,0.487907) (3,0.231148) (4,0.463368) (5,0.166433) (6,0.487333) (7,0.192417) (8,0.315857) (9,0.344458) (10,0.543224) (11,0.587309) (12,0.462365) (13,0.867096) (14,1.05644) (15,1.36354) (16,2.47039) (17,3.31201) (18,6.55595) (19,29.131) (20,12.7972) (21,1.74989) (22,1.74989) };
    \end{axis}
  \end{tikzpicture}
  \caption{%
    Distribution of relative dimensions over the 37 instances from Table~\ref{table_instances} except for \texttt{p2756} (avoiding a bias due to many failures dimension computations).
    A cut of dimension $k$ in an instance having $\ell$ cuts and with $\dim(P) = d$ contributes $1/(36 \ell)$ to its bar.
    For $k = -1$ (resp.\ $k=d$), this is $\emptyset$ (resp.\ $\infty$), and it is $k/(d-1)$ otherwise.
  }
  \label{fig_dimension_distribution}
\end{figure}

The first bar in Fig.~\ref{fig_dimension_distribution} indicates that for an instance chosen uniformly at random among the ones we considered and then a randomly chosen cut for this instance, this cut is non-supporting with probability greater than 35\,\%.
This is remarkably high and thus we conclude this paper by proposing to investigate means to (heuristically) test for such a situation, with the goal of strengthening a non-supporting cutting plane by a reduction its right-hand side.

\paragraph{Acknowledgments.}
We thank R.\ Hoeksma and M.\ Uetz as well as the \SCIP{} development team, in particular A.~Gleixner, C.~Hojny and M.~Pfetsch for valuable suggestions on the computational experiments and their presentation.

\newpage

\bibliographystyle{plain}
\bibliography{cut-dimensions}

\newpage

\appendix

\section{Additional plots}
\label{appendix_plots}

Here we provide additional instance-specific plots.
This underlines the conclusions drawn in Section~\ref{sec_computations} and allows inspection of
results for instances with certain characteristics (see Table~\ref{table_instances}).

\input{appendix_plots.tex}

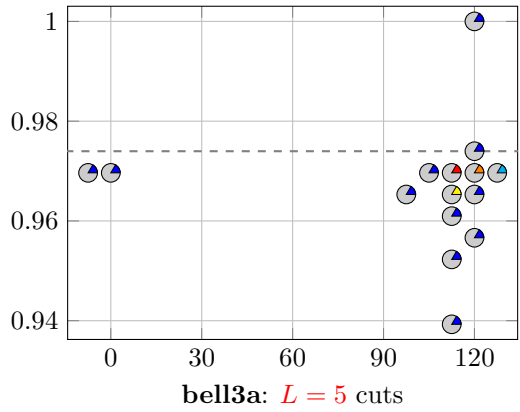
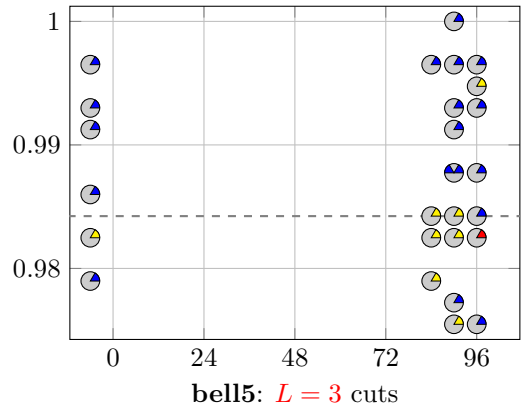
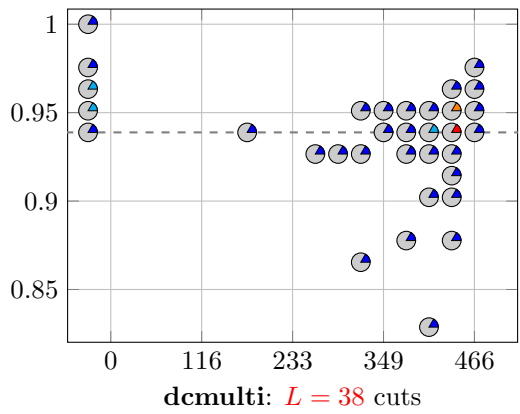
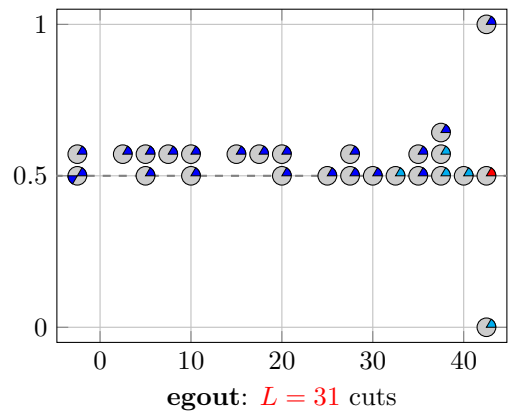
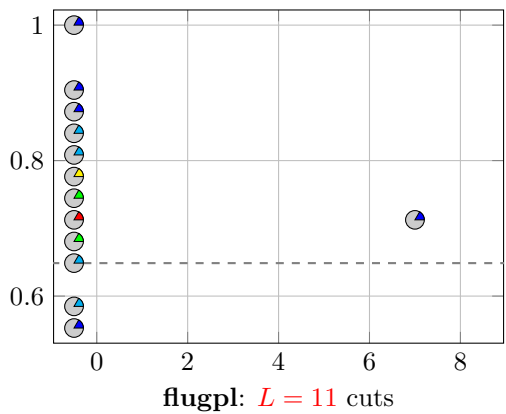
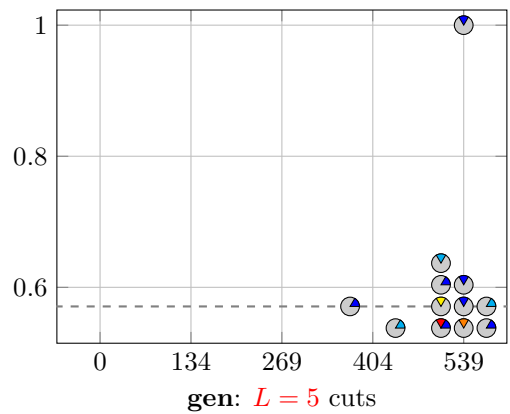
\end{document}